\documentclass[a4paper,UKenglish,cleveref,nameinlink,autoref,thm-restate]{lipics-v2021}

\usepackage{microtype}
\usepackage{xcolor,soul}
\usepackage{mdframed}
\usepackage{sidecap}
\sidecaptionvpos{figure}{c}

\usepackage{wrapfig}

\newif\iflong
\newif\ifshort

\newcommand{\shortLongAlternartive}[2]{\ifshort{#1}\else{#2}\fi}
\usepackage{apxproof}

\usepackage[utf8]{inputenc}
\usepackage{graphicx}
\usepackage{xspace}
\usepackage[utf8]{inputenc}
\usepackage{xspace}
\usepackage{complexity}
\usepackage{amsmath}
\usepackage{amssymb}
\usepackage{amsthm}
\usepackage{microtype}
\usepackage{array}
\usepackage{tabularx}
\usepackage{xcolor}
\usepackage{apxproof}
\usepackage{lineno}
\usepackage{float}
\usepackage{comment}
\usepackage[T1]{fontenc}
\usepackage{enumerate}  
\usepackage{boldline}

\DeclareMathOperator{\skel}{sk}

\definecolor{lipicsblue}{rgb}{0.08235294118,0.3098039216,0.537254902}
\hypersetup{
    colorlinks=true,
    linkcolor=lipicsblue,
    filecolor=lipicsblue,      
    urlcolor=lipicsblue,
    citecolor=lipicsblue,
    pdfpagemode=FullScreen,
    }

\newcommand{\col}[1]{\textcolor{lipicsblue}{#1}}

\ccsdesc[500]{Theory of computation~Graph algorithms analysis}
\ccsdesc[500]{Mathematics of computing~Graph algorithms}
\ccsdesc[500]{Human-centered computing~Graph drawings}

\renewcommand{\emph}[1]{\textcolor{blue}{\em #1}\xspace}

\usepackage{hyperref}
\newtheorem{obs}[theorem]{Observation}

\newcommand{\tabitem}{~~\llap{\textbullet}~~}

\usepackage{todonotes}
\usepackage{amsmath,amssymb}
\usepackage{apxproof}
\makeatother

\Crefname{observation}{Observation}{Observations}
\Crefname{algorithm}{Algorithm}{Algorithms}
\Crefname{algocf}{Algorithm}{Algorithms}
\Crefname{section}{Section}{Sections}
\Crefname{lemma}{Lemma}{Lemmata}
\Crefname{lemma2}{Lemma}{Lemmata}
\Crefname{note}{Note}{Notes}
\Crefname{claim}{Claim}{Claims}
\Crefname{property}{Property}{Properties}
\Crefname{enumi}{Condition}{Conditions}
\Crefname{figure}{Fig.}{Figs.}

\newtheoremrep{theorem2}[theorem]{Theorem}
\newtheoremrep{corollary2}[theorem]{Corollary}
\newtheoremrep{lemma2}[theorem]{Lemma}
\newtheoremrep{claim2}[theorem]{Claim}
\newtheoremrep{property2}[theorem]{Property}

\nolinenumbers

\title{Simple Realizability of Abstract Topological Graphs}
\author{Giordano {Da Lozzo}}{Roma Tre University, Italy \and \url{http://www.dia.uniroma3.it/~dalozzo}}{giordano.dalozzo@uniroma3.it}{https://orcid.org/0000-0003-2396-5174}{}
\author{Walter Didimo}{University of Perugia, Italy\and \url{https://mozart.diei.unipg.it/didimo/}}{walter.didimo@unipg.it}{https://orcid.org/0000-0002-4379-6059}{}
\author{Fabrizio Montecchiani}{University of Perugia, Italy\and \url{https://mozart.diei.unipg.it/montecchiani/}}{fabrizio.montecchiani@unipg.it}{https://orcid.org/0000-0002-0543-8912}{}
\author{Miriam Münch}{University of Passau, Germany\and \url{https://www.fim.uni-passau.de/theoretische-informatik/lehrstuhlteam/miriam-muench}}{muenchm@fim.uni-passau.de}{https://orcid.org/0000-0002-6997-8774}{}
\author{Maurizio Patrignani}{Roma Tre University, Italy\and \url{https://compunet.ing.uniroma3.it/\#!/people/titto}}{maurizio.patrignani@uniroma3.it}{https://orcid.org/0000-0001-9806-7411}{}
\author{Ignaz Rutter}{University of Passau, Germany\and \url{https://www.fim.uni-passau.de/en/theoretical-computer-science/chair/prof-dr-ignaz-rutter}}{rutter@fim.uni-passau.de}{https://orcid.org/0000-0002-3794-4406}{}
\date{}
\authorrunning{Da Lozzo et al.}
\keywords{Abstract Topological Graphs, SPQR-Trees, Synchronized PQ-Trees}
\Copyright{G. Da Lozzo, W. Didimo, F. Montecchiani, M. Münch, M. Patrignani, I. Rutter}
\funding{Research by {Da Lozzo}, Didimo, Montecchiani, and {Patrignani} was supported, in part,  by MUR of Italy (PRIN Project no.~2022ME9Z78 ~-- NextGRAAL and PRIN Project no.~2022TS4Y3N~-- EXPAND). 
Research by M\"unch and Rutter was funded by the Deutsche Forschungsgemeinschaft (DFG, German Research Foundation) -- 541433306.}
\acknowledgements{Research started at the~$1^{st}$ Summer Workshop on Graph Drawing, September 2021, Castiglione del Lago, Italy.}

\hideLIPIcs

\newcommand{\Chi}{\mathcal{X}\xspace}

\newcommand{\SATor}{{\sc SATR}\xspace} 
\newcommand{\ATRfull}{{\sc AT-graph Realizability}\xspace}
\newcommand{\SATRfull}{{\sc Simple AT-graph Realizability}\xspace}
\newcommand{\WATRfull}{{\sc Weak AT-graph Realizability}\xspace}
\newcommand{\ATR}{{\sc ATR}\xspace}
\newcommand{\SATR}{{\sc SATR}\xspace}
\newcommand{\WATR}{{\sc WATR}\xspace}
\newcommand{\threesat}{{\sc 3-SAT}\xspace}
\newcommand{\plthreesat}{{\sc Planar 3-SAT}\xspace}
\newcommand{\pthreesat}{{\sc 3-Connected Planar 3-SAT}\xspace}
\newcommand{\lcc}{\ensuremath{\mathrm{\lambda}}}
\newcommand{\acp}{{\sc ACP}\xspace}
\newcommand{\acpB}{{\sc 2-connected ACP}\xspace}
\newcommand{\acpG}{{\sc GACP}\xspace}
\newcommand{\acpGB}{{\sc 2-connected GACP}\xspace}

\soulregister{\acpG}{7}
\soulregister{\acpGB}{7}

\begin{document}

\maketitle

\begin{abstract}
An \emph{abstract topological graph} (\emph{AT-graph}) is a pair~$A=(G,\Chi)$, where~$G=(V,E)$ is a graph and~$\Chi \subseteq \binom{E}{2}$ is a set of pairs of edges of $G$. 
A \emph{realization of~$A$} is a drawing $\Gamma_A$ of $G$ in the plane such that any two edges~$e_1,e_2$ of~$G$ cross in~$\Gamma_A$ if and only if~$(e_1,e_2) \in \Chi$; $\Gamma_A$ is \emph{simple} if any two edges intersect at most once (either at a common endpoint or at a proper crossing).
The \emph{\ATRfull} (\emph{\ATR}) problem asks whether an input AT-graph admits a realization. The version of this problem that requires a simple realization is called \emph{\SATRfull} (\emph{\SATR}).
It is a classical result that both \ATR and \SATR are \NP-complete~\cite{DBLP:journals/jct/Kratochvil91a,Kratochvil1989}.

In this paper, we study the \SATR problem from a new structural perspective. More precisely, we consider the size $\lcc(A)$ of the largest connected component of the \emph{crossing graph} of any realization of $A$, i.e., the graph ${\cal C}(A) = (E, \Chi)$. This parameter represents a natural way to measure the level of interplay among edge crossings. First, we prove that \SATR is \NP-complete when $\lcc(A) \geq 6$. On the positive side, we give an optimal linear-time algorithm that solves \SATR when $\lcc(A) \leq 3$ and returns a simple realization if one exists. 
Our algorithm is based on several ingredients, in particular the reduction to a new embedding problem subject to constraints that require certain pairs of edges to alternate (in the rotation system), and a sequence of transformations that exploit the interplay between alternation constraints and the SPQR-tree and PQ-tree data structures to eventually arrive at a simpler embedding problem that can be solved with standard techniques.
\end{abstract}

\section{Introduction}\label{se:intro}
A \emph{topological graph} $\Gamma$ is a geometric representation of a graph $G=(V, E)$ in the plane such that the vertices of $G$ are mapped to distinct points and the edges of $G$ are simple curves connecting the points corresponding to their end-vertices. For simplicity, the geometric representations of the elements of $V$ and $E$ in $\Gamma$ are called \emph{vertices} and \emph{edges} of $\Gamma$, respectively. It is required that: $(i)$ any intersection point of two edges in $\Gamma$ is either a common endpoint or a crossing (a point where the two edges properly cross); $(ii)$ any two edges of $\Gamma$ have finitely many intersections and no three edges pass through the same crossing point. 
Additionally, we say that $\Gamma$ is \emph{simple} if adjacent edges never cross and any two non-adjacent edges cross at most once.
The \emph{crossing graph} ${\cal C}(\Gamma)$ of a topological graph $\Gamma$ is a graph whose vertices correspond to the edges \mbox{of $\Gamma$ (and hence of~$G$)} and two vertices are adjacent if and only if their corresponding \mbox{edges cross in $\Gamma$.}

An \emph{abstract topological graph} (\emph{AT-graph}) is a pair~$A=(G,\Chi)$ such that~$G=(V, E)$ is a graph and~$\Chi \subseteq \binom{E}{2}$ is a set of pairs of edges of $G$. We say that~$A$ is \emph{realizable} if there exists a topological graph~$\Gamma_A$ isomorphic to~$G$ such that any two edges $e$ and $e'$ of~$G$ cross (possibly multiple times) in~$\Gamma_A$ if and only if~$(e,e') \in \Chi$. The topological graph $\Gamma_A$ is called a \emph{realization of~$A$}. Note that, by definition, $A$ is realizable if and only if the crossing graph of any realization $\Gamma_A$ of $A$ is isomorphic to the graph $(E,\Chi)$. Since such a crossing graph only depends on $A$ (i.e., it is the same for any realization of $A$), we denote it by ${\cal C}(A)$.

The \emph{\ATRfull} (\emph{\ATR}) problem asks whether an AT-graph~$A=(G,\Chi)$ is realizable. The \emph{\SATRfull} (\emph{\SATR}) problem is the version of \ATR in which the realization of $A$ is required to be simple; if such a realization exists, then $A$ is said to be \emph{simply realizable}. 
Since the introduction of the concept of AT-graphs~\cite{DBLP:journals/siamdm/KratochvilLN91}, establishing the complexity of the \ATR (and of the \SATR) problem has been the subject of an intensive research activity, also due to its connection with other prominent problems in topological and geometric graph theory. Clearly, if $\Chi = \emptyset$, both the \ATR and the \SATR problems are equivalent to testing whether~$G$ is planar, which is solvable in linear time~\cite{DBLP:journals/jcss/BoothL76,DBLP:journals/jacm/HopcroftT74}. However, for $\Chi \neq \emptyset$, a seminal paper by Kratochv\'il~\cite{DBLP:journals/jct/Kratochvil91a} proves that \ATR is \NP-hard and that this problem is polynomially equivalent to recognizing \emph{string graphs}. We recall that a graph $S$ is a string graph if there exists a system of curves (called \emph{strings}) in the plane whose crossing graph is isomorphic to $S$.
In the same paper, Kratochv\'il proves the \NP-hardness of the \emph{\WATRfull} (\emph{\WATR}) problem, that is, deciding whether a given AT-graph $A = (G, \Chi)$ admits a realization where a pair of edges may cross only if it belongs to $\Chi$.
He also proves that recognizing string graphs remains polynomial-time reducible to \WATR. Subsequent results focused on establishing decision algorithms for \WATR; it was first proven that this problem belongs to NEXP~\cite{DBLP:journals/jcss/SchaeferS04} and then to NP~\cite{sss-rsgnp-03}. This also implies the \NP-completeness of recognizing string graphs (which is polynomial-time reducible to \WATR) and of \ATR (which is polynomially equivalent to string graph recognition).  
Concerning the simple realizability setting for AT-graphs, it is known that the \SATR problem remains \NP-complete, still exploiting the connection with recognizing string graphs~\cite{Kratochvil1989,DBLP:journals/jct/KratochvilM94}. On the positive side, for those AT-graphs $A=(G,\Chi)$ for which~$G$ is a complete $n$-vertex graph, \SATR is solvable in polynomial-time, with an~$O(n^6)$-time algorithm~\cite{DBLP:journals/dcg/Kyncl11,DBLP:journals/dcg/Kyncl20}. Refer to~\cite{DBLP:journals/dcg/Kyncl11} for the complexity of other variants of ATR, and to~\cite{DBLP:conf/wg/GassnerJPSS06} for a connection with the popular {\sc Simultaneous Graph Embedding} problem.

\subparagraph{Our contributions.} In this paper, we further investigate the complexity of the simple realizability setting, i.e., of the \SATR problem. We remark that focusing on simple drawings is a common scenario in topological graph theory, computational geometry, and graph drawing (see, 
e.g.,~\cite{DBLP:journals/csur/DidimoLM19,DBLP:reference/cg/2004,DBLP:books/sp/20/Hong20,DBLP:reference/crc/StreinuBPDCKAF99}), because avoiding crossings between adjacent edges, as well 
as multiple crossings between a pair of non-adjacent edges, is a requirement for minimal edge crossing layouts. 
Specifically, we study the simple realizability problem for an AT-graph $A$ from a new structural perspective, namely looking at  the number of vertices of the largest connected component of the crossing graph ${\cal C}(A)$, which we denote by $\lcc(A)$. This parameter is a natural measure of the level of interplay among edge crossings. Namely, \SATR is trivial on instances for which $\lcc(A) \le 2$, that is, instances in which the number of crossings is unbounded but each edge is crossed at most once. On the other hand, the problem becomes immediately nontrivial as soon as $\lcc(A) \geq 3$.  
Precisely, our results are as follows:

\begin{itemize}
    \item We prove that \SATRfull is \NP-complete already for instances~$A$ for 
 which~$\lcc(A) = 6$ (which, in fact, implies the hardness for every fixed value of $\lcc(A) \geq 6$); see \cref{se:hardness}. A consequence of our result is that, unless $\P = \NP$, the problem is not fixed-parameter tractable with respect to $\lcc(A)$ and, thus, with respect to any graph parameter bounded by $\lcc(A)$, such as the maximum node degree, the treewidth or even the treedepth. As the results in~\cite{DBLP:journals/jct/Kratochvil91a,Kratochvil1989,DBLP:journals/jct/KratochvilM94}, our hardness proof uses a reduction from \pthreesat. However, the reduction in~\cite{DBLP:journals/jct/Kratochvil91a} does not deal with the simplicity of the realization, whereas the reduction in \cite{Kratochvil1989,DBLP:journals/jct/KratochvilM94} may lead to instances~$A$ for which $\lcc(A)$ is greater than six and, actually, is even not bounded by~a~constant.
	
   \item We prove that \SATRfull can be solved efficiently \mbox{when~$\lcc(A) \leq 3$.} More precisely, we give an optimal $O(n)$-time testing algorithm, which also finds a simple realization if one exists; see \cref{se:triplets}. We remark that the only polynomial-time algorithm previously known in the literature for the \SATR problem is restricted to complete graphs and has high complexity~\cite{DBLP:journals/dcg/Kyncl11,DBLP:journals/dcg/Kyncl20}. Our algorithm is based on several ingredients, including the reduction to a new embedding problem subject to constraints that require certain pairs of edges to alternate (in the rotation system), and a sequence of transformations that exploit the interplay between alternation constraints and the SPQR-tree~\cite{DBLP:journals/algorithmica/BattistaT96,dt-opl-96} and PQ-tree~\cite{DBLP:journals/jcss/BoothL76,booth1975pq} data structures to eventually arrive at a simpler embedding problem that can be solved with standard techniques. We remark that the alternation constraints we encounter in our problem are rather opposite to the more-commonly encountered consecutivity constraints~\cite{DBLP:journals/talg/BlasiusFR23,BlasiusR16,DBLP:journals/jgaa/GutwengerKM08,Schaefer13} and cannot be handled straightforwardly with~PQ-trees.
    \end{itemize}
    
\smallskip

\section{Basic Definitions and Tools}\label{se:basic}

For basic definitions about graphs and their drawings, refer to~\cite{DBLP:books/ph/BattistaETT99,DBLP:books/daglib/0030488}.
We only consider simple realizations and thus we often omit the qualifier ``simple'' when clear from the~context.

Let $A=(G,\Chi)$ be an AT-graph, with $G=(V,E)$, and let~$\Gamma_A$ be a realization of~$A$. A \emph{face} of~$\Gamma_A$ is a region of the plane bounded by maximal uncrossed portions of the edges in~$E$.
A set~$E' \subseteq E$ of~$k$ edges pairwise crossing in~$\Gamma_A$ is a \emph{$k$-crossing}of~$\Gamma_A$.
As we focus on simple realizations, we assume that the edges in~$E'$ are pairwise non-adjacent in~$G$.
\mbox{For a~$k$-crossing~$E'$,} denote by~$V(E')$ the set of~$2k$ endpoints of the~$k$ edges in~$E'$. The \emph{arrangement} of~$E'$, denoted by~$C_{E'}$, is the arrangement of the curves \mbox{representing the edges of~$E'$ in~$\Gamma_A$.} 
A~$k$-crossing~$E'$ is \emph{untangled} if, in the arrangement~$C_{E'}$, all~$2k$ vertices in~$V(E')$ are incident to a common face (see \cref{fig:3crossing-untangled}); otherwise, it is \emph{tangled} (see \cref{fig:3crossing-tangled}). \mbox{The next lemma
will turn useful in \cref{se:triplets}.}

\begin{lemma}\label{lem:untangling}
An AT-graph $A$ with $\lcc(A) \leq 3$ admits a simple realization if and only if it admits a simple realization in which all~$3$-crossings are untangled.
\end{lemma}

\begin{proof}
The sufficiency is obvious, as any simple realization of an AT-graph in which all $3$-crossings are untangled is a specific simple realization. For the necessity,
let $A$ be an AT-graph with $\lcc(A) \leq 3$ and let 
$\Gamma_A$ be a simple realization of $A$ that contains a tangled $3$-crossing~$E'$. Since the edges in $E'$ pairwise cross and $\lcc(A) \leq 3$, the vertices of ${\cal C}(A)$ corresponding to the edges in $E'$ induce a maximal connected component of ${\cal C}(A)$ that is a 3-cycle (the edges in $E'$ do not cross any edge in $E \setminus E'$).
We show how to obtain a new simple realization $\Gamma'_A$ of $A$ that coincides with $\Gamma_A$ except for the drawing of one of the edges in $E'$ and such that $E'$ is untangled. Repeating such a transformation for each tangled $3$-crossing eventually yields the desired simple realization of $A$ with no tangled $3$-crossings.

Consider the arrangement $C_{E'}$ of $E'$ in $\Gamma_A$. Since $\Gamma_A$ is simple and $|E'| = 3$, we have that $C_{E'}$ splits the plane into two faces, one bounded and one unbounded. 
Let $e_1 = (u_1,v_1)$,~$e_2 = (u_2,v_2)$, and $e_3 = (u_3,v_3)$ be the edges in $E'$.
Since $E'$ is tangled, one of such faces, denoted by $f$, contains an endpoint of two edges of $E'$ and the other face, denoted by $h$, contains the remaining four endpoints. In the remainder, we assume that $f$ is bounded (refer to \cref{fig:3crossing}). This is not a loss of generality, as a stereographic projection can be applied to $\Gamma_A$ to satisfy this assumption. 
W.l.o.g., we also assume that the endpoints of the edges of $E'$ incident to $f$ are $v_1$ and $v_2$.
Denote by $G_f$ (resp.\ by $G_h$) the subgraph of $G$ consisting of the vertices and of the edges of $G$ that lie in the interior of $f$ (resp.\ of $h$); refer to \cref{fig:3crossing}.
Let $Q$ be a closed curve passing through $v_1$ and $v_2$ that encloses the drawing of $G_f$ in $\Gamma_A$, without intersecting any other vertex or edge; refer to \cref{fig:3crossing-untangled}.
Let us follow the drawing of the edge $e_2$ from $u_2$ to $v_2$. Suppose that, as in \cref{fig:3crossing}, at the intersection point $p$ between $e_2$ and $e_1$, we see the endpoint $u_1$ of $e_1$ to the left; the case in which $u_1$ is to the right being symmetric. 
Consider three points $a$, $b$, and $c$ along $e_1$ in $\Gamma_A$, encountered in this order going from $u_1$ to~$v_1$, such that $a$ and $b$ lie arbitrarily close to $p$, and
$c$ lies arbitrarily close to $v_1$. 
To obtain $\Gamma'_A$, replace the drawing of $e_1$ in $\Gamma_A$ with the union of four curves $\lambda_1,\lambda_3,\lambda_3,\lambda_4$ defined as follows (refer to \cref{fig:3crossing-curves}): 
\begin{itemize}
\item $\lambda_1$ coincides with the drawing of $e_1$ in $\Gamma_A$ between $u_1$ and $a$ (it only intersects $e_3$);
\item $\lambda_2$ starts at $a$, follows the drawing of $e_2$ between $p$ and $v_2$, and then follows clockwise the border of $Q$ between $v_2$ and $v_1$ until it reaches $c$ (it intersects no edge);
\item $\lambda_3$ coincides with the drawing of $e_1$ between $c$ and $b$ (it only intersects $e_2$); and
\item $\lambda_4$ starts at $b$, follows the drawing of $e_2$ between $p$ and $v_2$, and then follows clockwise the border of $Q$ between $v_2$ and $v_1$ until it reaches $v_1$ (it intersects no edge).
\end{itemize}
Since the drawing of $\Gamma'_A$ coincides with $\Gamma_A$ except for the drawing of $e_1$, and since we have redrawn $e_1$ in such a way that it only intersects (once) $e_2$ and $e_3$, we have that $\Gamma'_A$ is a simple realization of $A$. Moreover, when traversing $e_2$ from $u_2$ to $v_2$, we now have that, at the intersection point between $e_2$ and $e_1$, we see the endpoint $u_1$ of $e_1$ to the right. Therefore, all the endpoints of the edges in $E'$ lie in the same face of the arrangement they form in $\Gamma'_A$, that is, $E'$ is untangled in $\Gamma'_A$.
\end{proof}

\begin{figure}[t!]
    \centering
    \begin{subfigure}[t]{.3\textwidth}
\includegraphics[page=1,width=\textwidth]{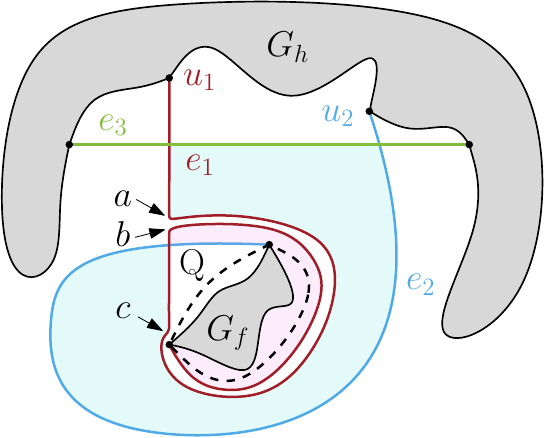}
    \caption{Realization $\Gamma_A$}
    \label{fig:3crossing-tangled}
    \end{subfigure}
    \hfil
    \begin{subfigure}[t]{.3\textwidth}
\includegraphics[page=2,width=\textwidth]{figs/3crossings.pdf}
\caption{Realization $\Gamma'_A$}
    \label{fig:3crossing-untangled}
    \end{subfigure}
    \hfil
    \begin{subfigure}[t]{.17\textwidth}
\includegraphics[page=3,width=\textwidth]{figs/3crossings.pdf}
\caption{Curves of $e_1$}
    \label{fig:3crossing-curves}
    \end{subfigure}
\caption{Illustrations for the proof of \cref{lem:untangling}. (a) A schematic representation of a simple realization $\Gamma_A$ of an AT-graph~$A$ with a tangled $3$-crossing $E' = \{e_1, e_2, e_3\}$. (b) The simple realization $\Gamma'_A$ obtained from $\Gamma_A$, where $E'$ is untangled.
(c) The curves forming $e_1$ in~$\Gamma'_A$.}
    \label{fig:3crossing}
\end{figure}

\subsection{Connectivity.}
Let $G$ be a graph. 
The graph $G$ is \emph{k-connected} if it contains at least $k$ vertices and the removal of any $k-1$ vertices leaves it connected; $2$- and $3$-connected graphs are also called \emph{biconnected} and \emph{triconnected}, respectively. 
A \emph{cut vertex} (resp. \emph{separation pair}) of $G$ is a vertex (resp. a pair of vertices) whose removal disconnects $G$. 
Therefore, the graph $G$ is biconnected (resp.\ triconnected) if it has no cut vertex (resp.\ no separation pair). 
A \emph{biconnected component} (or \emph{block}) of~$G$ is a maximal (in terms of vertices and edges) biconnected subgraph of~$G$. 
The \emph{cut components} of a cut vertex $v$ of $G$ are the subgraphs of $G$ induced by $v$ together with the maximal subsets of the vertices of $G$ that are not disconnected by the removal of $v$.

\subsection{PQ-trees.}
PQ-trees were first introduced by Booth and Lueker~\cite{booth1975pq,DBLP:journals/jcss/BoothL76}.
An \emph{unrooted PQ-tree} over a finite set~$U$ is an unrooted tree whose leaves are the elements of~$U$ and
whose internal nodes are labeled as P- or Q-nodes.
A P-node is usually depicted
as a circle, a Q-node as a rectangle; see \cref{fig:PQ-examples}.
Consider an unrooted PQ-tree~$T$ and let $L$ be the circular order of its leaves. 
Then~$T$ represents the set of those circular orders of its leaves that can be obtained from~$L$ by~$(i)$ 
arbitrarily permuting the children of arbitrarily many P-nodes and~$(ii)$ reversing the order of arbitrarily
many Q-nodes. Since all PQ-trees we consider in the paper are unrooted, we will refer to~\emph{unrooted PQ-trees} 
if we write PQ-trees.
A~\emph{synchronized} PQ-tree contains a pair of Q-nodes that have to be flipped simultaneously; see~\cref{fig:PQ-examples}\col{$(c)$}.

\begin{figure}[h!]
    \centering
    \includegraphics{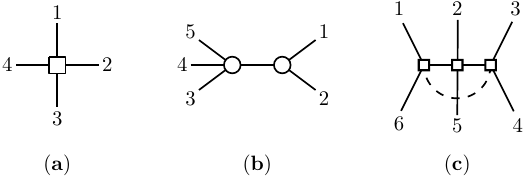}
    \caption{PQ-trees representing the circular orders $(a)$ $(1, 2, 3, 4)$ and its reverse, $(b)$ of $1, 2, 3, 4, 5$, such that $1, 2$ are consecutive, $(c)$ $(6, 1, 2, 3, 4, 5)$, $(6, 1, 5, 3, 4, 2)$, $(1, 6, 2, 4, 3, 5)$, $(1, 6, 5, 4, 3, 2)$.}
    \label{fig:PQ-examples}
\end{figure}

\subsection{SPQR-trees.}
A \emph{split pair} of a graph~$G$ is either a pair of adjacent vertices or a \emph{separation pair}, i.e., a pair of vertices whose removal disconnects~$G$.
The \emph{split components}~of~$G$ \emph{with respect to} a split pair $\{u,v\}$ are defined as follows.
If $(u,v)$ is an edge of $G$, then it is a split component of $G$ with respect to $\{u,v\}$.
Also, let $G_1,\dots,G_k$ be the connected components of $G \setminus \{u,v\}$. The subgraphs of $G$ induced by $V(G_i) \cup \{u,v\}$, minus the edge $(u,v)$, are components of $G$ with respect to $\{u,v\}$, for $i=1,\dots,k$.
The SPQR-tree~$T$ of~a biconnected graph~$H$, introduced by Di Battista and
Tamassia~\cite{DBLP:journals/algorithmica/BattistaT96,dt-opl-96}, represents a decomposition of~$H$ into triconnected components along its split pairs.
Each node~$\mu$ of~$T$ contains a graph~$\skel(\mu)$, called \emph{skeleton of~$\mu$}. The edges of~$\skel(\mu)$ are either edges of~$H$, which we call \emph{real edges}, or newly introduced edges, called \emph{virtual edges}.
The tree~$T$ is initialized to a single node~$\mu$, whose skeleton, composed only of real edges, is~$H$. 
Consider a split pair~$\{u, v\}$ of the skeleton of some node~$\mu$ of~$T$,
and let~$H_1,\dots,H_k$ be the split components of~$H$ with respect to~$\{u,v\}$ such that~$H_1$ is not a virtual edge and, if~$k=2$, also~$H_2$ is not a virtual edge.  
We introduce a node~$\nu$ adjacent to~$\mu$ whose skeleton is the graph~$H_1 + e_{\nu,\mu}$, where~$e_{\nu,\mu} = (u,v)$ is a virtual edge, and replace~$\skel(\mu)$ with the graph~$\bigcup_{i \neq 1} H_i + e_{\mu,\nu}$, where~$e_{\mu,\nu} = (u,v)$ is a virtual edge. We say that $e_{\nu,\mu}$ is the \emph{twin virtual edge} of $e_{\mu,\nu}$
and vice versa.
Applying this replacement iteratively produces a tree with more nodes but smaller skeletons associated with the nodes. Eventually, when no further replacement is possible, the skeletons of the nodes of~$T$ are of four types: parallels of at least three virtual edges ($P$-nodes), parallels of exactly one virtual edge and one real edge ($Q$-nodes), cycles ($S$-nodes), and triconnected planar graphs ($R$-nodes).  See \Cref{fig:spqr} for an example. The two vertices in the skeleton of an $P$-node~$\mu$ are called the \emph{poles} of~$\mu$.

The \emph{merge} of two adjacent nodes~$\mu$ and~$\nu$ in~$T$, replaces~$\mu$ and~$\nu$ in~$T$ with a new node~$\tau$ that is adjacent to all the neighbors of~$\mu$ and~$\nu$, and whose skeleton is~$\skel(\mu) \cup \skel(\nu) \setminus \{ e_{\mu,\nu}, e_{\nu,\mu}\})$, where the end-vertices of~$e_{\mu,\nu}$ and~$e_{\nu,\mu}$ that correspond to the same vertices of~$H$ are identified. By iteratively merging adjacent~$S$-nodes, we eventually obtain the (unique) SPQR-tree data structure as introduced by Di Battista and Tamassia~\cite{DBLP:journals/algorithmica/BattistaT96,dt-opl-96}, where the skeleton of an~$S$-node is a cycle. The crucial property of this decomposition is that a planar embedding of~$H$ uniquely induces a planar embedding of the skeletons of its nodes and that, arbitrarily and independently, choosing planar embeddings for all the skeletons uniquely determines an embedding of~$H$. Observe that the skeletons of~$S$- and~$Q$-nodes have a unique planar embedding, that the skeleton of~$R$-nodes have two planar embeddings (where one is the reverse of the other), and that~$P$-nodes have as many planar embedding as the permutations of their virtual edges.
If~$H$ has~$n$ vertices, then~$T$ has~$O(n)$ nodes and the total number of virtual edges in the skeletons of the nodes of~$T$ is in~$O(n)$.  From a computational complexity perspective,~$T$ can be constructed in~$O(n)$ time~\cite{DBLP:conf/gd/GutwengerM00}.

\begin{figure}[tb]
    \centering
    \includegraphics{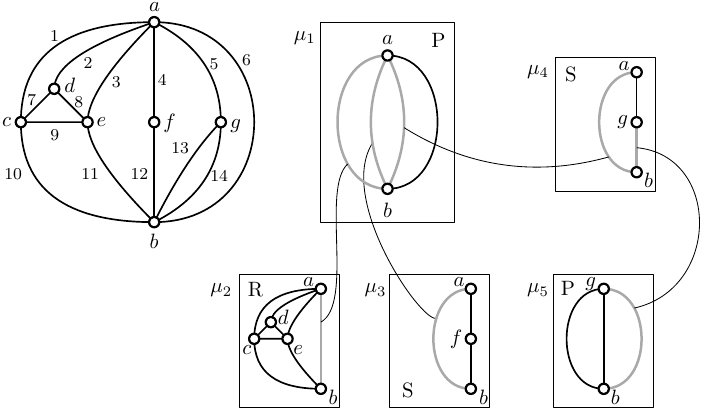}
    \caption{A graph~$H$ and its SPQR-tree decomposition (virtual edges drawn in gray), where Q-nodes are omitted in favor of allowing real edges (drawn in black) in the skeletons of the other nodes.}
    \label{fig:spqr}
\end{figure}

\section{NP-completeness for AT-Graphs with~\texorpdfstring{$\mathbf{\lcc(A) \geq 6}$}{lcc(A) >= 6}}\label{se:hardness}

In this section, we show that the \SATor problem is \NP-complete for an AT-graph~$A$ even when the largest component of the crossing graph~${\cal C}(A)$ has bounded size; specifically, when $\lcc(A) = 6$ (see \cref{th:hardness}). 
We will exploit a reduction from the \NP-complete problem \pthreesat~\cite{k-spspc-94}.

\noindent
 
Let~$\varphi$ be a Boolean formula in conjunctive normal form. The \emph{variable-clause graph}~$G_\varphi$ \emph{of}~$\varphi$ is the bipartite (multi-)graph that has a node for each variable and for each clause, and an edge between a variable-node and a clause-node if a positive or negated literal of the variable appears in the clause.
If each clause of~$\varphi$ has exactly three literals corresponding to different variables and~$G_\varphi$ is planar and triconnected, then~$\varphi$ is an instance of \pthreesat. Observe that in this case $G_\varphi$ is a simple graph. Our proof exploits several gadgets described hereafter, which are then combined to obtain the desired reduction. 

Intuitively, in the instance $A_\varphi$ of \SATor corresponding to $\varphi$, we encode truth values into the clockwise or counter-clockwise order in which some edges cross suitable cycles of the subgraphs (called ``variable gadgets'') representing the variables of $\varphi$ in $A_\varphi$. These edges connect the variable gadgets to the subgraphs (called ``clause gadgets'') representing those clauses that contain a literal of the corresponding variable. Only if at least one of the literals appearing in a clause gadget encodes a \texttt{true} value, the clause gadget admits a simple realization. We start by describing the ``skeleton'' of $A_\varphi$, that is the part of $A_\varphi$ that encloses all the gadgets and ensures that they are properly connected. Next, we describe the ``split gadget'' which, in turn, is used to construct the variable gadget. If a variable $v$ has literals in $k$ clauses, we have $k$ pairs of edges leaving the corresponding variable gadget and entering the $k$ clause gadgets. The clause gadgets always receive three truth values, corresponding to the three literals of the corresponding clause.

\smallskip

\subparagraph{Skeleton.} 
Arbitrarily choose a planar embedding~$\mathcal E_\varphi$ of~$G_\varphi$.
The \emph{skeleton}~$H_\varphi$ \emph{of}~$\varphi$ is a~$4$-regular~$3$-connected plane graph obtained from~$\mathcal E_\varphi$ as follows; refer to \cref{fig:skeleton}\shortLongAlternartive{ in \cref{apx:se:hardness}}{}. For each degree-$k$ variable~$v$ of~$\varphi$, the graph~$H_\varphi$ contains a~$k$-cycle formed by the sequence of edges~$(e_{v,c_1}, e_{v,c_2},\dots,e_{v,c_k})$, which we refer to as the \emph{variable cycle} of~$v$, where~$c_1,\dots,c_k$ are the clause-nodes of~$G_\varphi$ adjacent to~$v$ in the clockwise order in which they appear around~$v$ in~$\mathcal{E}_\varphi$.\linebreak
For each clause~$c$ of~$\varphi$, the graph~$H_\varphi$ contains a~$3$-cycle formed by the sequence of edges~$(e_{c,v_1}, e_{c,v_2},e_{c,v_3})$, which we refer to as the \emph{clause cycle} of~$c$, where~$v_1,v_2,v_3$ are the variable-nodes of~$G_\varphi$ adjacent to~$c$ in the clockwise order in which they appear around~$c$ in~$\mathcal{E}_\varphi$.\linebreak
For each edge~$(v_i,c_j)$ of~$G_\varphi$, the graph~$H_\varphi$ contains two edges, which we 
refer to as the \emph{pipe edges} of the edge~$(v_i,c_j)$,
connecting the endpoints of~$e_{v_i,c_j}$ and~$e_{c_j,v_i}$ without crossings.
We need the following lemma to show that the skeleton~$H_\varphi$ obtained from~$\mathcal{E}_\varphi$ is triconnected.

\begin{lemma}\label{lem:3-paths}
Let $G$ be a triconnected graph, and let $u$ and $v$ be two vertices of $G$. Also, let~$e_u$ and $e_v$ be two edges incident to $u$ and to $v$, respectively. Then, $G$ contains three vertex disjoint paths connecting $u$ and $v$ that use both $e_u$ and $e_v$.
\end{lemma}

\begin{proof}
Recall that, by Menger's theorem~\cite{Menger1927}, since $G$ is triconnected, it contains three vertex-disjoint paths $P_1$, $P_2$, and $P_3$ connecting $u$ and $v$. If such paths use both $e_u$ and $e_v$, there is nothing to prove. Thus, suppose that at least one of $e_u$ and $e_v$, w.l.o.g., say $e_u$, does not appear in any of such paths. We distinguish two cases: 
{\bf Case 1}: $e_v$ belongs to one of~$P_1$,~$P_2$, and $P_3$; and
{\bf Case 2}:  $e_v$ does not belong to any of $P_1$, $P_2$, and $P_3$.

In {\bf Case 1}, assume w.l.o.g. that $e_v$ belongs to $P_1$.
Consider a path $P_u$ between $u$ and~$v$ that contains $e_u$. If $P_u$ is disjoint, except at its endpoints, from all other paths, then~$P_u$,~$P_1$, and any of $P_2$ and $P_3$ satisfy the requirements of the statement, since such paths are vertex-disjoint and since $e_u \in E(P_u)$ and $e_v \in E(P_1)$.
If $P_u$ intersects one of the other paths at an internal vertex, then let $x$ be the first vertex of $P_u$, encountered when traversing this path from $u$ to $v$, that belongs to a path $P_i$, with $i \in \{1,2,3\}$.
Let $P_x$ be the path obtained by concatenating the subpath of $P_u$ between $u$ and $x$ and the subpath of $P_i$ between $x$ and $v$. Then, 
if $i=1$, the path $P_x$ contains both $e_u$ and $e_v$, and it is vertex disjoint from $P_2$ and $P_3$. Thus, $P_x$, together with $P_2$ and $P_3$ satisfy the requirements of the statement.
Otherwise, the path $P_x$ together $P_1$ and the path $P_j$ with $j \in \{1,2,3\} \setminus \{i,1\}$ satisfy the requirements of the statement.

In {\bf Case 2}, consider a path $P_u$ between $u$ and $v$ that contains $e_u$.
We show how to modify paths $P_1$, $P_2$, and $P_3$ so that the conditions of {\bf Case 1} apply to $P_1$, $P_2$, and $P_3$ (up to renaming $u$ and $v$) and we can thus proceed as described above.
If $P_u$ is disjoint, except at its endpoints, from all other paths, then we set $P_1 = P_x$. 
If $P_u$ intersects one of the other paths at an internal vertex, we define the vertex $x$ and the path $P_x$ as above. Let $i$ be the index in~$\{1,2,3\}$ of the path $P_i$ containing $x$. Then, we set $P_i = P_x$. Since $P_x$ contains~$e_u$, {\bf Case 1} holds for the updated paths $P_1$, $P_2$, and $P_3$ (up to renaming $u$ and $v$). This concludes the proof.
\end{proof}

\begin{lemma}\label{lem:skeleton}
The skeleton~$H_\varphi$ obtained from~$\mathcal{E}_\varphi$ is triconnected.
\end{lemma}

\begin{proof}
\begin{figure}[t!]
    \centering
    \begin{subfigure}{.3\textwidth}
    \includegraphics[page=8,width=\textwidth]{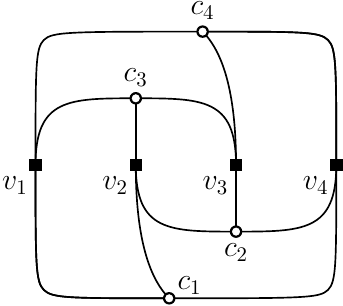}
    \end{subfigure}
    \hfill
    \begin{subfigure}{.6\textwidth}
    \includegraphics[page=7,width=\textwidth]{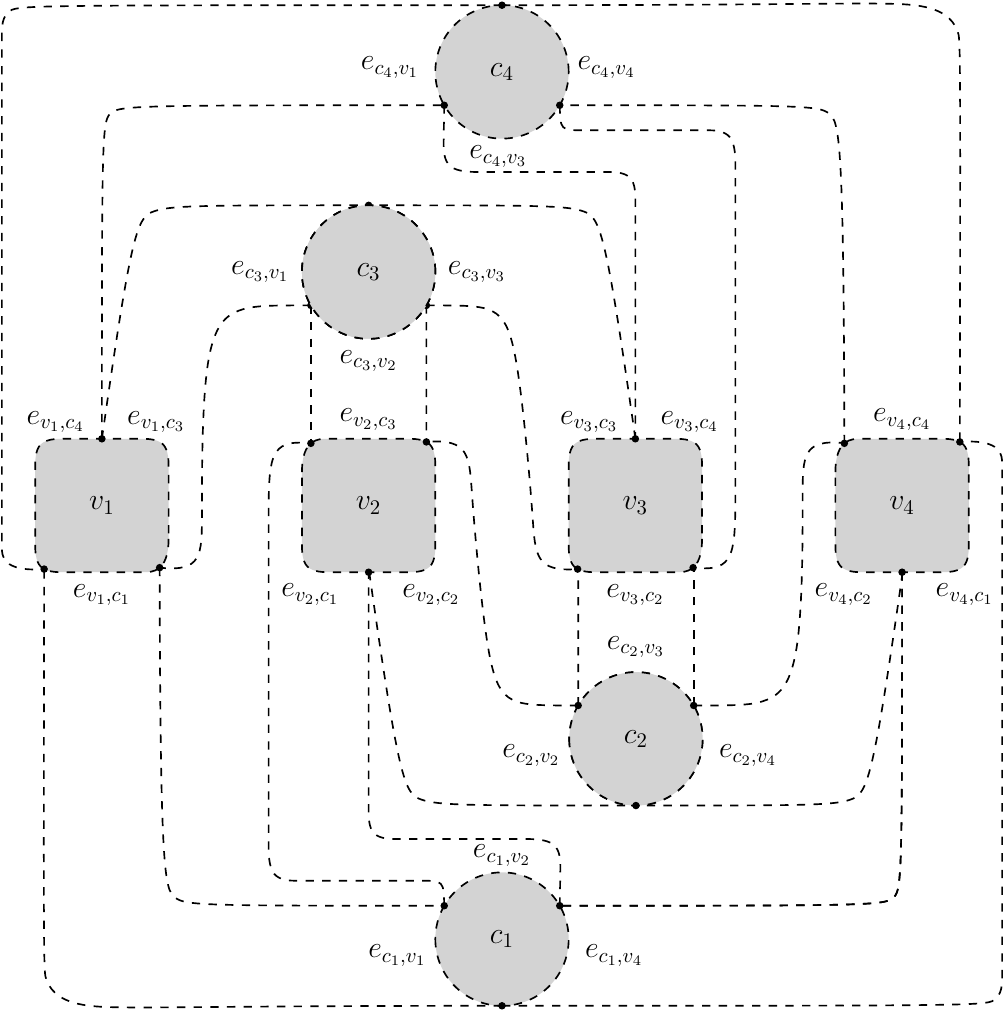}
    \end{subfigure}
    \caption{(Left) A planar $3$-connected variable-clause graph~$G_\varphi$ and (Right) the corresponding skeleton~$H_\varphi$.}
    \label{fig:skeleton}
\end{figure}
Recall that, by Menger's theorem~\cite{Menger1927}, a graph is triconnected if and only if it contains three vertex-disjoint paths between any two vertices, and that, by construction, each vertex of~$H_\varphi$ belongs either to a variable cycle or to a clause cycle stemming from a vertex of~$G_\varphi$. Therefore, in order to prove the triconnectivity of $G_\varphi$, we show the existence of three vertex-disjoint paths between any two vertices $u$ and~$v$ of~$H_\varphi$. 

Consider first the case in which~$u$ and~$v$ belong to the same variable or clause cycle~$C$ of $H_\varphi$. Then there exist already two vertex-disjoint paths between~$u$ and~$v$ along~$C$. Also, since removing the vertex corresponding to~$C$ from~$G_\varphi$ leaves the resulting graph connected, we also have that removing from~$H_\varphi$ the edges of~$C$ together with the vertices of~$C$ different from $u$ and $v$ leaves the resulting sub-graph of~$H_\varphi$ connected. This implies that there exists in~$G_\varphi$ a vertex-disjoint path between~$u$ and~$v$ that avoids the edges of~$C$ and the remaining vertices of~$C$.

Consider now the case that~$u$ and~$v$ belong to distinct variable or clause cycles of $H_\varphi$, which we denote by~$C_1$ and~$C_2$, respectively. The cycles $C_1$ and $C_2$ stem from a node $n_1$ and a node $n_2$ of $G_\varphi$, respectively. Let $e_u$ (resp.\ $e_v$) be one of the two edges of $H_\varphi$ incident to $u$ (resp.\ to $v$) and not belonging to $C_1$ (resp.\ to $C_2$). Observe that $e_u$ and $e_v$ are pipe edges of some edge $e_1$ of $G_\varphi$ (incident to $n_1$) and of some edge $e_2$ of $G_\varphi$ (incident to $n_2$), respectively. By carefully choosing $e_u$ and $e_v$ we can enforce that $e_1 \neq e_2$. 
It follows that $e_u$ and $e_v$ are never the two pipe edges of the same edge of $G_\varphi$.
By \cref{lem:3-paths}, there exist three vertex-disjoint paths $P_1$, $P_2$, and $P_3$ in $G_\varphi$ between $n_1$ and $n_2$ that use $e_1$ and $e_2$. We show that there exist three vertex-disjoint paths in $H_\varphi$ between $u$ and $v$. 
Let $P_i$, $i \in \{1,2,3\}$ be one such path. For each edge $e$ of $P_i$ we select an edge $e_H$ of $H_\varphi$ such that $e_H$ is a pipe edge of $e$: If $e$ is incident to $n_1$ (resp. to $n_2$), $e_H = e_u$ (resp. $e_H = e_v$); otherwise, $e_H$ is arbitrarily chosen. 
It is possible to construct a path $p_i$ in $H_\varphi$ composed of the edges of $H_\varphi$ defined for~$P_i$, where for each pair of edges $e'_H$ and $e''_H$ originating from two edges of $G_\varphi$ incident to the same vertex $n_j$, we add the edges of any of the two paths of $H_\varphi$ connecting the endpoints of $e'_H$ and $e''_H$ along $C_j$. The construction of $p_i$ fails only in the case in which the two edges $e_u$ and $e_v$ are the two pipe edges of the same edge $(n_1,n_2)$, which is ruled out by the way we choose $e_u$ and $e_v$. We have obtained three vertex-disjoint paths $p_1$, $p_2$, and~$p_3$ in $H_\varphi$ from vertices of $C_1$ to vertices of $C_2$ that do not use edges of $C_1$ and $C_2$ and such that~$u$ and $v$ are both extremes of some of these paths. By using the edges of $C_1$ and $C_2$ it is possible to construct three vertex-disjoint paths in $H_\varphi$ between $u$ and $v$.
\end{proof}

\begin{figure}[t!]
        \begin{subfigure}[t]{.3\textwidth}
\includegraphics[page=2,scale=.6]{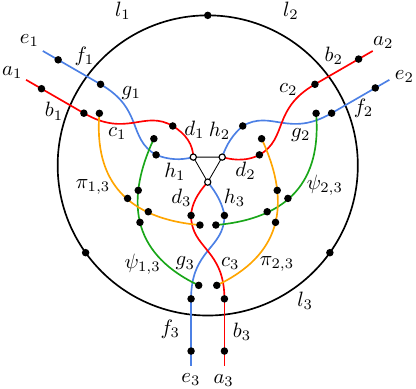}
    \caption{}
    \label{fig:splitter-a}
    \end{subfigure}
    \hfil
\begin{subfigure}[t]{.3\textwidth}
\includegraphics[page=3,scale=.6]{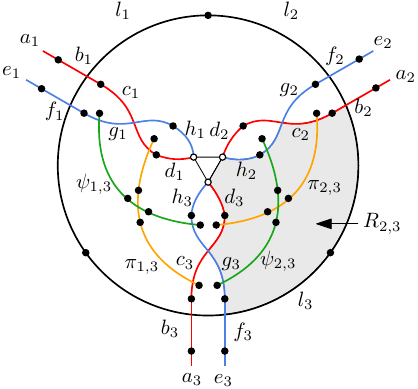}
    \caption{}
    \label{fig:splitter-b}
\end{subfigure}
\hfil
\begin{subfigure}[t]{.3\textwidth}
\includegraphics[page=5,scale=.48]{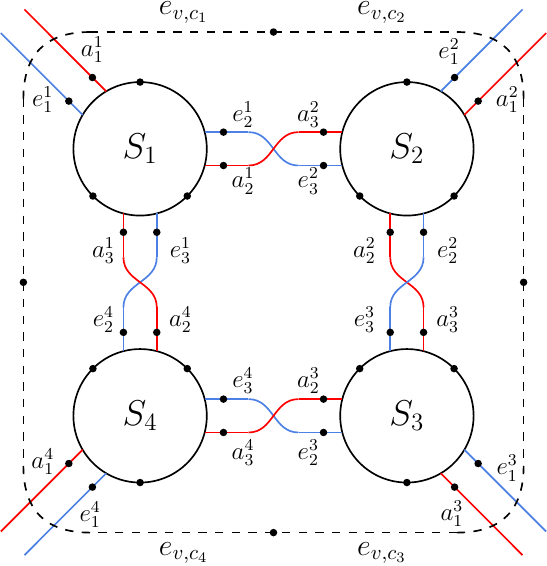}
    \caption{}
    \label{fig:variable}
    \end{subfigure}    
    \caption{(a,b) The split gadget~$S$ : The clockwise circular order of the edges leaving the gadget is either~$b_1$,~$f_1$,~$b_2$,~$f_2$,~$b_3$,~$f_3$ (a) or~$f_1$,~$b_1$,~$f_2$,~$b_2$,~$f_3$,~$b_3$  (b).
    (c) The variable gadget~$\mathcal{V}_v$. The dashed edges belong to the variable~cycle~of~$v$ in the skeleton $H_\varphi$.
    } 
\end{figure}

\subparagraph{Split gadget.} The split gadget~$S$ is the AT-graph defined as follows; refer to \cref{fig:splitter-a,fig:splitter-b}. The underlying graph of~$S$ consists of six connected components:  
(1) a~$3$-cycle formed by the edges~$l_1$, $l_2$, and $l_3$, which we call \emph{outer cycle} of~$S$; 
(2) a~$3$-cycle formed by the vertices~$v_1$, $v_2$, and $v_3$ (filled white in \cref{fig:splitter-a,fig:splitter-b}) such that, for~$i=1,2,3$, the vertex~$v_i$ is the endpoint of the two paths formed by the sequence of edges~$(a_i,b_i,c_i,d_i)$ (red paths in \cref{fig:splitter-a,fig:splitter-b}) and~$(e_i,f_i,g_i,h_i)$ (blue paths in \cref{fig:splitter-a,fig:splitter-b});
(3) four length-$3$ paths~$\pi_{1,3}$,~$\pi_{2,3}$,~$\psi_{1,3}$, and~$\psi_{2,3}$. 
We denote the first, intermediate, and last edge of a length-$3$ path~$p$ as~$p'$,~$p''$, and~$p'''$, respectively.
The crossing graph~${\cal C}(S)$ of~$S$ consists of several connected components. Next, we describe the eight \emph{non-trivial connected components} of~${\cal C}(S)$, i.e., those that are not isolated vertices, determined by the following crossings of the edges of~$S$ (see \cref{fig:components-variable}):
(i) For~$j=1,2,3$, edge~$l_j$ crosses both~$b_j$ and~$f_j$;
(ii) For~$j=1,2$, edge~$\pi''_{j,3}$ crosses~$\psi''_{j,3}$;
(iii) For~$j=1,2$, edge $c_j$ crosses~$g_j$ and $\pi'_{j,3}$, further $g_j$ crosses~$\psi'_{j,3}$;
finally (iv) edge~$c_3$ crosses~$g_3$,~$\pi'''_{1,3}$, and~$\pi'''_{2,3}$, further~$g_3$ crosses~$\psi'''_{1,3}$ and~$\psi'''_{2,3}$.

\begin{figure}[t!]
    \centering
    \begin{subfigure}{\textwidth}
    \centering
    \includegraphics[page=6,scale=.8]{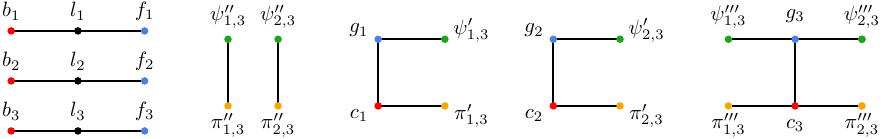}
    \caption{Non-trivial connected components of $\mathcal{C}(S)$.}
    \label{fig:components-variable}
    \end{subfigure}
    
    \begin{subfigure}{\textwidth}
    \centering
    \includegraphics[page=9,scale=.8]{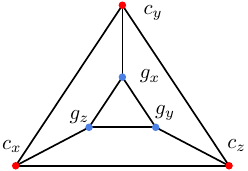}
    \caption{The non-trivial connected component of $\mathcal{C}(\mathcal{Q}_c)$.}
    \label{fig:components-clause}
    \end{subfigure}
    
    \begin{subfigure}{\textwidth}
    \centering
    \includegraphics[page=12,scale=.8]{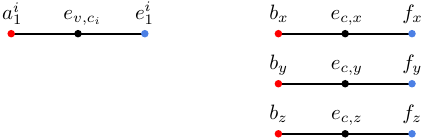}
    \caption{
    The non-trivial connected components of $\mathcal{C}(A_\varphi)$ involving an edge of the skeleton.}
    \label{fig:components-skeleton}
    \end{subfigure}
\caption{Illustration of the non-trivial connected components of (a) the crossing graph of the split gadget $S$ and (b) the crossing graph of the clause gadget $\mathcal{Q}_c$. (c) The non-trivial connected components of $\mathcal{C}(A_\varphi)$ involving an edge of the variable cycle (left) and of the clause cycle (right).
}
    \label{fig:components-crossing-graph}
\end{figure}

\begin{lemma}
\label{lem:splitter}
    In any simple realization of $S$, the circular (clockwise or counterclockwise) order of the edges crossing the outer cycle of~$S$ is either~$b_1$,~$f_1$,~$b_2$,~$f_2$,~$b_3$,~$f_3$ or~$f_1$,~$b_1$,~$f_2$,~$b_2$,~$f_3$,~$b_3$.
\end{lemma}

\begin{proof}
Let~$C$ be the outer cycle of~$S$, which consists of the edges~$l_1, l_2,$ and~$l_3$. Consider a simple realization~$\Gamma$ of~$S$, such that ~$l_1$,~$l_2$, and~$l_3$ are encountered in this order when clockwise traversing~$C$ (refer to \cref{fig:splitter-a,fig:splitter-b}); the case in which these edges are encountered in the reverse order in~$\Gamma$ being symmetric.
    Recall that~$v_i$ is the vertex shared by the edges~$h_i$ and~$d_i$, for~$i=1,2,3$, and let~$\Delta$ be the triangle induced by~$v_1$,~$v_2$, and~$v_3$ in the underlying graph of~$S$. Observe that, the edges of~$\Delta$ are not involved in any crossings. Therefore, in~$\Gamma$, the triangle~$\Delta$ must lie either in the interior or in the exterior of~$C$. W.l.o.g., we assume that~$\Delta$ lies in the interior of~$C$, the case in which~$\Delta$ lies in the exterior of~$C$ being symmetric.
    It follows that, for~$i=1,2,3$, the two paths composed of the edges~$h_i, g_i$ and the edges~$d_i,c_i$, respectively, lie in the interior of~$C$. This is due to the fact that these paths are incident to~$v_i$ and that their edges do not cross the edges of~$C$. This, in turn, implies that the paths~$\pi_{i,3}$ and~$\psi_{i,3}$, with~$i \in \{1,2\}$, also lie in the interior of~$C$.

    Consider the drawing~$\Gamma'$ obtained by removing from~$\Gamma$ the paths~$\pi_{i,3}$ and~$\psi_{i,3}$, with~$i \in \{1,2\}$. 
    We show two properties of $\Gamma'$: (a) $\Delta$ bounds an internal face of~$\Gamma'$ and (b) $v_1$, $v_2$, and $v_3$ appear in this order when clockwise traversing~$\Delta$.
    To prove Property~(a), recall that $\Delta$ lies inside $C$ by hypothesis and that its edges cannot be crossed. Moreover, since the edges $b_i$ and $f_i$ must cross the edge $l_i$, for $i = 1,2,3$, the paths $(b_i,c_i,d_i)$ and~$(f_i,g_i,h_i)$ must be drawn in $\Gamma'$ in the exterior of $\Delta$. Therefore, $\Delta$ bounds an internal face of~$\Gamma'$.
    To prove Property~(b), recall that $l_1$,~$l_2$, and~$l_3$ are encountered in this order when clockwise traversing~$C$ by hypothesis and that the edge $b_i$ crosses the edge $l_i$, for $i= 1,2,3$.
    Therefore, since the paths $b_i,c_i,d_i$ and $b_j,c_j,d_j$, with $i \neq j$, do not cross each other, we have that $v_1$, $v_2$, and $v_3$ must appear in this order when clockwise traversing~$\Delta$. 
    Properties (a) and (b) allow us to define the following regions.
    Denote by~$R_{1,2}$ the bounded region of the plane incident to the edge~$(v_1,v_2)$ and to the vertex shared by~$l_1$ and~$l_2$.  Analogously, denote by~$R_{2,3}$ the bounded region of the plane incident to the edge~$(v_2,v_3)$ and to the vertex shared by~$l_2$ and~$l_3$ (see the gray region in \cref{fig:splitter-a,fig:splitter-b}), and denote by~$R_{1,3}$ the bounded region of the plane incident to the edge~$(v_3,v_1)$ and to the vertex shared by~$l_3$ and~$l_1$. 
    
    To prove the statement it suffices to show that when clockwise traversing~$C$ the following sequences of edges may not occur
    (i)~$b_2$,~$f_2$,~$f_3$, and~$b_3$;
    (ii)~$f_2$,~$b_2$,~$b_3$, and~$f_3$;
    (iii)~$b_1$,~$f_1$,~$f_3$, and~$b_3$;
    (iv)~$f_1$,~$b_1$,~$b_3$, and~$f_3$.
    Suppose, for a contradiction, that case (i) occurs; the other cases being analogous. Since~$\Gamma$ (and thus~$\Gamma'$) is simple and since~$c_i$ crosses~$g_i$, for~$i=2,3$, we have that the edges~$g_3, c_3, d_3$,~$(v_3,v_2)$,~$d_2, c_2, g_2$ are encountered in this order when clockwise traversing the boundary of~$R_{2,3}$. Consider now the two paths~$\pi_{2,3}$ and~$\psi_{2,3}$ in~$\Gamma$. Since the edges~$\pi''_{2,3}$ and~$\psi''_{2,3}$ cross in~$\Gamma$, we have that  the endpoints of~$\pi_{2,3}$ alternate with the endpoints of~$\psi_{2,3}$ in an Eulerian clockwise tour of the drawing of~$\pi_{2,3}  \cup \psi_{2,3}$ in~$\Gamma$. This, in particular, implies that when clockwise traversing the boundary of~$R_{2,3}$ either the
    edges~$\pi'_{2,3}$,~$\psi'_{2,3}$,~$\pi'''_{2,3}$, and~$\psi'''_{2,3}$ are encountered in this order or in the order~$\psi'_{2,3}$,~$\pi'_{2,3}$,~$\psi'''_{2,3}$, and~$\pi'''_{2,3}$.
    Recall that, in~$\Gamma$, the edge~$\pi'_{2,3}$, $\psi'_{2,3}$, $\pi'''_{2,3}$, and~$\psi'''_{2,3}$ must cross the edge~$c_2$, $g_2$, $c_3$, and~$g_3$, respectively. However, this is impossible due to the fact that~$g_3$,~$c_3$,~$c_2$, and~$g_2$ appear in this order when clockwise traversing~$R_{2,3}$.
\end{proof}

\subparagraph{Variable gadget.}
For each variable~$v$ of degree~$k$ in~$\varphi$ and incident to clauses~$c_1, c_2, \dots, c_k$ in~$G_\varphi$, the \emph{variable gadget}~$\mathcal{V}_v$ is an AT-graph defined as follows; refer to \cref{fig:variable}. Assume, w.l.o.g., that~$c_1, c_2, \dots, c_k$ appear in this clockwise order around~$v$ in~$\mathcal E_\varphi$. 
The underlying graph of~$\mathcal{V}_v$ is composed of~$k$ split gadgets~$S_1, S_2, \dots, S_k$. For each split gadget~$S_i$, with~$i = 1, \dots, k$, rename the edges~$a_j$ and~$e_j$ of~$S_i$ as~$a^i_j$ and~$e^i_j$, respectively, with~$j \in \{1,2,3\}$. For~$i=1,\dots,k$, we identify the edges~$a^i_j$ and~$a^{i+1}_{j+1}$ and the edges~$e^i_j$ and~$e^{i+1}_{j+1}$, where~$k+1=1$. 
The crossing graph~${\cal C}(\mathcal{V}_v)$ of~$\mathcal{V}_v$ consists of all vertices and edges of the crossing graphs of~$S_i$, with~$i=1, \dots, k$. Moreover, for $i=1,\dots,k$, it contains a non-trivial connected component consisting of the single edge~$(a^i_2,e^i_2)$ (which coincides with~$(a^{i+1}_3,e^{i+1}_3)$,~$i+1 = 1$ when~$i=k$).

\begin{lemma}
 \label{lem:variable}
    In any simple realization of~$\mathcal{V}_v$ together with the variable cycle of $v$ in which both~$a^i_1$ and~$e^i_1$ cross~$e_{v,c_i}$, for~$i=1,\dots,k$, the clockwise circular order of the edges crossing the variable cycle 
    of~$v$ in $\mathcal{V}_v$ is either~$a^1_1$,~$e^1_1$,~$a^2_1$,~$e^2_1$,\dots,~$a^k_1$,~$e^k_1$ or~$e^1_1$,~$a^1_1$,~$e^2_1$,~$a^2_1$,\dots,~$e^k_1$,~$a^k_1$.
\end{lemma}

\begin{proof}
Let~$C$ be the variable cycle of~$v$ in $H_\varphi$, which consists of the edges~$e_{v,c_i}$, for~$i=1,\dots,k$. Consider a simple realization~$\Gamma$ of~$\mathcal{V}_v \cup C$ in which both~$a^i_1$ and~$e^i_1$ cross~$e_{v,c_i}$, for~$i=1,\dots,k$, and the edges~$(e_{v,c_1},e_{v,c_2}, \dots, e_{v,c_k})$ appear in this order when clockwise traversing~$C$ (refer to \cref{fig:variable}); the case in which these edges are encountered in the reverse order in~$\Gamma$ being symmetric. 

Consider the graph~$\mathcal{V}^-_v: = \mathcal{V}_v \setminus \bigcup^k_{i=1} \{a^i_1, e^i_1\}$.
We have that, for any two edges~$e$ and~$e'$ of~$\mathcal{V}^-_v$, there exists a chain~$r_1,\dots,r_h$ of edges of~$\mathcal{V}^-_v$, such that~$r_1=e$,~$r_h=e'$, and, for~$i=1,\dots,h-1$, it holds that either~$r_i$ and~$r_{i+1}$ are adjacent or~$r_i$ and~$r_{i+1}$ cross. This property and the fact that no edge of~$\mathcal{V}^-_v$ crosses an edge of~$C$ implies that~$\mathcal{V}^-_v$ lies either entirely in the interior or entirely in the exterior of~$C$ in~$\Gamma$. W.l.o.g, we assume that~$\mathcal{V}^-_v$ lies in the interior of~$C$, the case in which it lies in the exterior of~$C$ being symmetric.
To prove the statement, it suffices to show that, for~$i=1,\dots,k-1$, either~$e^i_1$,~$a^i_1$,~$e^{i+1}_1$,~$a^{i+1}_1$ appear in this clockwise order around~$C$ or such edges appear the clockwise order~$a^i_1$,~$e^i_1$,~$a^{i+1}_1$,~$e^{i+1}_1$ around~$C$.

Consider a pair of consecutive split gadget~$S_i$ and~$S_{i+1}$, for~$i=1,\dots, k-1$, and let~$\Gamma'$ be the drawing of the graph~$\mathcal{N}:= S_i \cup S_{i+1}$ 
in~$\Gamma$.
By \cref{lem:splitter}, for~$i=1,\dots,k$, the edges~$e^i_1$,~$a^i_1$,~$e^i_2$,~$a^i_2$,~$e^i_3$, and~$a^i_3$ appear in this clockwise order or in the clockwise order~$a^i_1$,~$e^i_1$,~$a^i_2$,~$e^i_2$,~$a^i_3$, and~$e^i_3$  around the outer cycle of~$S_i$.
This leaves us with four possible combinations of orders for the edges exiting~$\mathcal{N}$. However,
since~$e^i_2=e^{i+1}_3$ crosses~$a^i_2=a^{i+1}_3$, we have that in a clockwise visit of the external face of~$\Gamma'$ only two combinations are feasible; refer to \cref{fig:variable}. Namely, we either encounter the endpoints of~$e^i_1$,~$a^i_1$,~$e^{i+1}_1$, and~$a^{i+1}_1$ in this order or the endpoints of~$a^i_1$,~$e^i_1$,~$a^{i+1}_1$, and~$e^{i+1}_1$ in this order. This concludes the proof.
\end{proof}

In the proof of \cref{th:hardness}, the two circular orders for the edges $\bigcup^k_{i=1} \{a^i_1,e^i_1\}$ of $\mathcal{V}_v$ considered in \cref{lem:variable} will correspond to the two possible truth assignment of the variable~$v$.

\subparagraph{Clause gadget.}
For each clause $c$ in $\varphi$, the \emph{clause gadget}~$\mathcal{Q}_c$ is the AT-graph, whose construction is inspired by a similar gadget used in~\cite{Kratochvil1989}, defined as follows; see \cref{fig:clause}. The underlying graph of~$\mathcal{Q}_c$ consists of six length-$3$ paths: For~$v \in \{x,y,z\}$, we have a path formed by the edges~$(a_v,b_v,c_v)$ (red paths in \cref{fig:clause}) and a path formed by the edges~$(e_v,f_v,g_v)$ (blue paths in \cref{fig:clause}).
The crossing graph~${\cal C}(\mathcal{Q}_c)$ of~$\mathcal{Q}_c$ consists of one non-trivial connected component (see \cref{fig:components-clause})
formed by the triangles~$(c_x,c_y,c_z)$ and~$(g_x,g_y,g_z)$, and the edges~$(c_x,g_z)$,~$(c_y,g_x)$, and~$(c_z,g_y)$.


\begin{figure}[tb!]
    \centering
        \centering
    \begin{subfigure}{.45\textwidth}
\includegraphics[page=3,width=\textwidth]{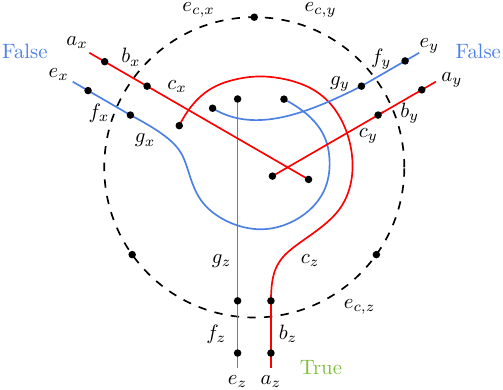}
\vspace{-5mm}
    \caption{}
    \label{fig:clause-true-1}
    \end{subfigure}
    \hfil
    \begin{subfigure}{.45\textwidth}
\includegraphics[page=4,width=\textwidth]{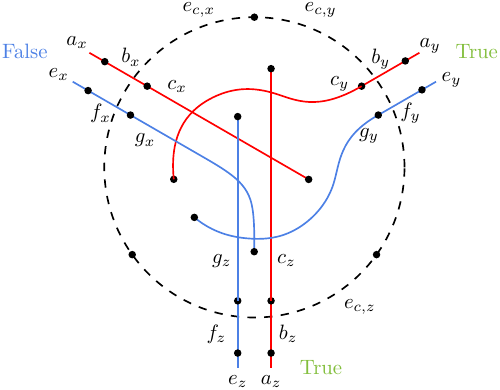}
\vspace{-5mm}
\caption{}
    \label{fig:clause-true-2}
    \end{subfigure}
\\
    \begin{subfigure}{.45\textwidth}
\includegraphics[page=5,width=\textwidth]{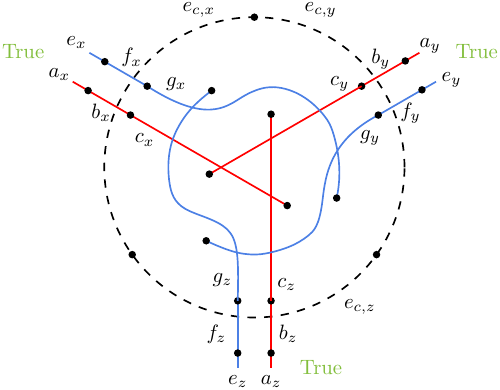}
\vspace{-5mm}
\caption{}
    \label{fig:clause-true-3}
    \end{subfigure}
\hfil
    \begin{subfigure}{.45\textwidth}
\includegraphics[page=7,width=\textwidth]{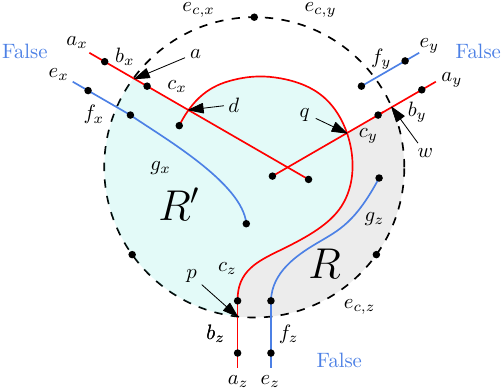}
\vspace{-5mm}
\caption{}
    \label{fig:clause-false}
    \end{subfigure}
    \caption{The clause gadget. Illustrations corresponding to a satisfied clause (a-c). Partial illustration of an unsatisfied clause (d).}
    \label{fig:clause}
\end{figure}

\begin{lemma}
\label{lem:clause-necessity}
    The clause gadget~$\mathcal{Q}_c$ admits a simple realization together with the clause cycle of $c$ in which, for~$v \in \{x,y,z\}$, both~$b_v$ and~$f_v$ {cross~$e_{c,v}$,} and in which the {edges~$e_{c,x}$,~$e_{c,y}$, and~$e_{c,z}$} appear in this order when traversing clockwise the clause cycle of~$c$ if and only if for at least one pair~$f_v, b_v$, with~$v \in \{x,y,z\}$, we have that~$b_v$ precedes~$f_v$ {along~$e_{c,v}$} \mbox{when traversing the clause cycle of~$c$~clockwise.}
\end{lemma}

\begin{proof}
\cref{fig:clause-true-1,fig:clause-true-2,fig:clause-true-3} show that, up to a renaming of the variables $x$, $y$, and $z$, if for at least one pair~$f_v, b_v$, with~$v \in \{x,y,z\}$, we have that~$b_v$ precedes~$f_v$ {along~$e_{c,v}$} when traversing the clause cycle of~$c$ clockwise, then
the clause gadget~$\mathcal{Q}_c$ admits a simple realization together with the clause cycle $c$ in which, for~$v \in \{x,y,z\}$, both~$b_v$ and~$f_v$ {cross~$e_{c,v}$, and in which the edges~$e_{c,x}$,~$e_{c,y}$, and~$e_{c,z}$} appear in this order when traversing clockwise the clause cycle of~$c$.

For the opposite direction, consider a simple realization $\Gamma$ of $\mathcal{Q}_c$ together with the clause cycle $c$ in which, for~$v \in \{x,y,z\}$, both~$b_v$ and~$f_v$ {cross~$e_{c,v}$, and in which the edges~$e_{c,x}$,~$l_{c,y}$, and~$l_{c,z}$} appear in this order when traversing clockwise the clause cycle of $c$. Also, assume that, in $\Gamma$, there exists no pair $f_v, b_v$, with~$v \in \{x,y,z\}$, such that~$b_v$ precedes~$f_v$ {along~$e_{c,v}$} when traversing the clause cycle of~$c$ clockwise; refer to \cref{fig:clause-false}.

Consider the region $R$ delimited by the following curves:
(i) the part of $b_z$ between its intersection with {$e_{c,z}$,} which we denote by $p$, and its common endpoint with $c_z$, 
(ii) the part of $c_z$ between such an endpoint and the intersection of $c_z$ with $c_y$, which we denote by $q$,
(iii) the part of $c_y$ between $q$ and its common endpoint with $b_y$,
(iv) the part of $b_y$ between such an endpoint and the intersection of $b_y$ with {$e_{c,y}$,} which we denote by $w$,
(v) the part of {$e_{c,y}$} between $w$ and its common endpoint with {$e_{c,z}$,} and
(vii) the part of {$e_{c,z}$} between such an endpoint and $p$.
Observe that the edges whose parts compose the boundary of $R$, that is, the edges {$e_{c,y}$, $e_{c,z}$,} $b_z$, $c_z$, $c_y$, and $b_y$ may not be crossed by $g_z$.
Also, since $f_z$ precedes $b_z$ along~{$e_{c,z}$} when traversing the clause cycle of~$c$ clockwise, we have that the common endpoint of $f_z$ and $g_z$ lies in $R$. It follows that the drawing of $g_z$ in $\Gamma$ lies in the interior of $R$.

Consider now the region $R'$ delimited by the following curves:
(i) the part of $b_z$ between~$p$ and its common endpoint with $c_z$, 
(ii) the part of $c_z$ between such an endpoint and the intersection of $c_z$ with $c_x$, which we denote by $d$,
(iii) the part of $c_x$ between $d$ and its common endpoint with $b_x$,
(iv) the part of $b_x$ between such an endpoint and the intersection of $b_x$ with~{$e_{c,x}$,} which we denote by $a$,
(v) the part of~{$e_{c,x}$} between $a$ and its common endpoint with~{$e_{c,z}$,} and
(vii) the part of {$e_{c,z}$} between its common endpoint with {$e_{c,x}$} and $p$.
Observe that the edges whose parts compose the boundary of $R'$, that is, {the edges $e_{c,x}$, $e_{c,z}$,} $b_z$, $c_z$,~$c_x$, and $b_x$ may not be crossed by $g_x$.
Also, since $f_x$ precedes $b_x$ along {$e_{c,x}$} when traversing the clause cycle of~$c$ clockwise, we have that the common endpoint of $f_x$ and $g_x$ lies in $R'$. It follows that the drawing of $g_x$ in $\Gamma$ lies in the interior of $R'$.

Since, by definition, $R$ and $R'$ are interior disjoint, and since we have shown that, in~$\Gamma$, the drawing of $g_z$ lies in $R$ and the drawing of $g_x$ lies in $R'$, it follows that $g_x$ and $g_z$ cannot cross in~$\Gamma$, which contradicts the fact that $\Gamma$ is a simple realization of $\mathcal{Q}_c$ satisfying the conditions of the statement.
\end{proof}

Based on \cref{lem:clause-necessity}, we associate the \texttt{True} value with a literal of a variable $v \in \{x,y,z\}$ appearing in $c$ when $b_v$ precedes~$f_v$ {along~$e_{c,v}$} while traversing the clause cycle of~$c$~clockwise, and \texttt{False} otherwise; see \cref{fig:clause}.
We can finally prove the main result of the section.

\begin{theorem}
\label{th:hardness}
    \SATor is \NP-complete for instances $A$ with $\lcc(A) = 6$. 
\end{theorem}

\begin{proof}
  We give a reduction from the \NP-complete problem \pthreesat~\cite{k-spspc-94}. Let $\varphi$ be an instance of \pthreesat.
  We construct an instance~$A_\varphi = (G',\Chi')$ of \SATor that is simply realizable if and only if $\varphi$ is satisfiable.
  We initialize $G' = H_\varphi$ and $\Chi' = \emptyset$.
  Then, for each variable $v$, we extend $A_\varphi$ to include~$\mathcal{V}_v$ as follows: For each clause $c_i$ that contains a literal of~$v$, add to $\Chi'$ the pair of edges $\{a^i_1,e_{v,c_i}\}$ and $\{e^i_1,e_{v,c_i}\}$, where  $a^i_1$ and $e^i_1$ belong to $\mathcal{V}_v$, and $e_{v,c_i}$ belongs to $H_\varphi$.
  Also, for each clause~$c$, we extend $A_\varphi$ to include $\mathcal{Q}_c$ as follows:
  For each variable $v \in \{x,y,z\}$ whose literals belong to $c$, we add to $\Chi'$ the pair of {edges $\{f_v, e_{c,v}\}$ and $\{b_v,e_{c,v}\}$,} where $f_v$ and $b_v$ belong to~$\mathcal{Q}_c$, and {$e_{c,v}$} belongs to $H_\varphi$.
  Finally, for each occurrence of a literal of a variable $v$ to a clause~$c_i$, we
  identify edges of $\mathcal{V}_v$ with edges of $\mathcal{Q}_v$ as follows: If $v$ appears as a positive (resp. negated) literal in $c_i$, then we identify the edge $a^i_1$ of $\mathcal{V}_v$ with the edge $a_y$ (resp.\ $e_y$) of $\mathcal{Q}_v$ and
  we identify the edge $e^i_1$ of $\mathcal{V}_v$ with the edge $e_y$ (resp.~$a_y$) of $\mathcal{Q}_v$.
  Observe that we do not allow the edges $a^i_1$  and $e^i_1$ to cross.
  Clearly, $A_\varphi$ can be constructed in polynomial time.
  The equivalence between $A_\varphi$ and $\varphi$ immediately follows from \cref{lem:variable,lem:clause-necessity}, and from the fact that, by \cref{lem:skeleton}, in any simple realization of $A_\varphi$, all the variable cycles and all the clause cycles maintain the same circular orientation.

  Observe that the size of the largest connected component of the crossing graph $\mathcal{C}(A_\varphi)$ is six; see also \cref{fig:components-crossing-graph}.
  As discussed at the start of the section, this implies \NP-hardness for any~$\lcc(A) = k$, for any $k \geq 6$.
  In particular, such bound is achieved by a connected component of $\mathcal{C}(S)$, for each split gadget $S$ used to construct $\mathcal{V}_v$, for any variable $v$, and by a connected component of $\mathcal{C}(\mathcal{Q}_c)$, for any clause $c$. The component of $\mathcal{C}(S)$ of size six is defined by the 
  edges $c_3$, $g_3$, $\pi'''_{1,3}$, $\psi'''_{1,3}$, $\pi'''_{2,3}$, and $\psi'''_{2,3}$ (see \cref{fig:splitter-a,fig:splitter-b}). 
  The component of $\mathcal{C}(\mathcal{Q}_c)$ of size six is defined by the 
  edges $c_r$ and $g_r$, with $r \in \{x,y,z\}$ (see \cref{fig:clause}). 

  Finally, to show that \SATor is in \NP, observe that one can test whether an instance~$A = (G,\Chi)$ admits a simple realization, by considering all possible orders of the crossings of the edges of $G$ that form the pairs in $\Chi$, run a linear-time planarity testing of the corresponding planarization of $G$, and return \texttt{YES} if and only if at least one such tests succeeds. This concludes the proof that \SATor is \NP-complete.
\end{proof}

We remark that the \NP-hardness of \cref{th:hardness} holds for instances whose crossing graph is planar, and has maximum degree~$3$ and treewidth~$3$; see \cref{fig:components-crossing-graph}. Moreover, it implies that \SATor is \NP-complete when~$\lcc(A) \geq k$, for any~$k \geq 6$. 
Finally, since our reduction yields instances whose size is linear in the size of the input (planar)~$3$-SAT formula, we have the following.

\begin{corollary}
Unless ETH fails, \SATor has no~$2^{o(\sqrt{n})}$-time algorithm, where $n$ is the number of vertices of the input AT-graph.
\end{corollary}

\begin{proof}
    By the {\em Sparsification Lemma}~\cite{DBLP:journals/jcss/ImpagliazzoPZ01}, it holds that \threesat cannot be solved in~$2^{o(N+M)}$ time, where~$N$ and~$M$ are the number of variables and clauses of the formula, respectively, unless the Exponential Time Hypothesis (ETH) fails.  On the other hand, there exists a polynomial-time reduction from \threesat to \plthreesat that transforms \threesat instances with~$N+M$ variables and clauses into equivalent \plthreesat instances with~$O((N+M)^2)$ variables and clauses~\cite{DBLP:journals/siamcomp/Lichtenstein82}.  Moreover, there is a reduction from \plthreesat to \pthreesat that only requires a linear blow-up in the size of the instance~\cite{k-spspc-94}. Altogether, given an instance of \threesat with~$N+M$ variables and clauses, there is a polynomial-time transformation that maps such instance into an equivalent instance of \pthreesat with~$O((N+M)^2)$ variables and clauses. Finally, starting from an instance of \pthreesat, our reduction constructs a graph whose number of vertices and edges is~$n+m \in O(n) = O((N+M)^2)$. Therefore, an algorithm solving \SATor in~$2^{o(\sqrt{n})}$ time would contradict ETH.
\end{proof}

\section{A Linear-Time Algorithm for AT-Graphs with~\texorpdfstring{$\mathbf{\lcc(A) \leq 3}$}{lcc(A) <= 3}}\label{se:triplets}
  
In this section we show that the problem \SATR can be solved in linear-time for AT-graphs~$A$ with~$\lcc(A) \leq~3$; see \cref{th:main-theorem}.
 We first give a short high-level overview of the overall strategy but note that proper definitions will only be given later in the detailed description of the algorithm.
The first step is to reduce \SATR to a constrained embedding problem where each vertex~$v$ may be equipped with alternation constraints that restrict the allowed orders of its incident edges around~$v$.
Next, we further reduce to the biconnected variant of the embedding problem which leads to new types of alternation constraints. It will turn out that many of these constraints 
can be transformed into constraints that can be expressed in terms of PQ-trees and are therefore easier to handle. Finally, we show that, when no further such transformations are possible, all the remaining alternation constraints have a simple structure that allows for an efficient test.

We now start with reducing \SATR to a constrained embedding problem.
Let~$A=(G,\Chi)$ be an~$n$-vertex AT-graph such that~$\lcc(A) \leq 3$. 
We construct from~$G$ an auxiliary graph~$H$ as follows.  
For each connected component~$X$ of~${\cal C}(A)$ that is not an isolated vertex, denote by~$E(X)$ the set of edges of~$G$ corresponding to the vertices of~$X$, and by~$V(X)$ the vertices of~$G$ that are end-vertices of the edges in~$E(X)$. Remove from~$G$ the edges in~$E(X)$ and add a \emph{crossing vertex}~$v_X$ adjacent to all vertices in~$V(X)$; see \cref{fig:crossing vertex}.
Since no two crossing vertices are adjacent, the graph~$H$ does not depend on the order in which we apply these operations.  The edges incident to~$v_X$ are partitioned into pairs of edges where two edges~$(a,v_X)$ and~$(b,v_X)$ form a pair if~$(a,b)$ is an edge of~$G$ corresponding to a vertex of~$X$.  We call~$(a,v_X)$ and~$(b,v_X)$ the \emph{portions} of~$(a,b)$ and say that~$(a,v_X)$ and~$(b,v_X)$ \emph{stem from}~$(a,b)$ . 
Note that since~$\lcc(A) \leq 3$, a crossing vertex~$v_X$ has either degree~$4$ or~$6$. In the first case~$X$ is a~$K_2$ and in the latter case~$X$ is either an induced~$3$-path~$P_3$ or a triangle~$K_3$.
If~$X = K_2$, we color its two vertices red and blue, respectively. If~$X = K_3$ we color its three vertices red, blue, and purple, respectively. If~$X = P_3$, its vertex of degree~$2$ is colored purple, whereas we color red and blue the remaining two vertices, respectively.
Based on \cref{lem:untangling}, we observe the following.

\begin{figure}[t]
    \centering
    \includegraphics[page=1,scale=.92]{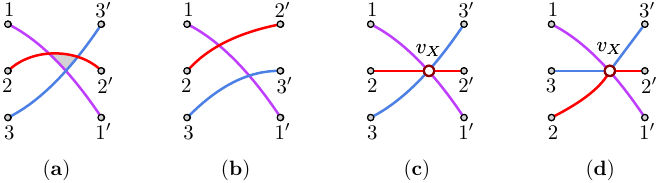}
    \caption{$(a)$ A~$K_3$-crossing.~$(b)$ A~$P_3$-crossing. A circular order of the neighbors around a crossing vertex~$v_X$ satisfying~$(c)$ a~$K_3$-constraint but not a~$P_3$-constraint,~$(d)$ a~$P_3$- but not a~$K_3$-constraint.}
    \label{fig:crossing vertex}
\end{figure}

\begin{obs}
\label{ob:h-is-planar}
    If~$A$ admits a simple realization, then~$H$ is planar.
\end{obs}

\begin{proof}
    By \cref{lem:untangling}, we can consider an embedding of~$G$ induced by a simple realization of~$A$ in which all~$3$-crossings are untangled. In this embedding, replace each crossing with a dummy vertex.  For components of~${\cal C}(A)$ of size~2, the dummy vertex corresponds to the crossing.  For components of~${\cal C}(A)$ of size~3, we distinguish whether the component is a path or a 3-cycle.  In the case of a path, we can contract the edge between the two dummy vertices to obtain the corresponding crossing vertex.  For a 3-cycle, by \cref{lem:untangling} we can further assume that the triangle enclosed by the edges connecting the three dummy vertices is empty, and we can therefore contract them into a single vertex.
\end{proof}

\cref{ob:h-is-planar} gives an immediate necessary condition for the realizability of~$A$, which is, however, not sufficient.  Indeed, a planar embedding of~$H$ obtained from contracting edges in a realization of~$A$ satisfies an additional property: 
for each crossing vertex~$v_X$ the portions stemming from two distinct edges~$e, f$ of~$G$ alternate around~$v_X$ {\em if and only} if the two vertices corresponding to~$e, f$ in~$X$ are adjacent. 
To keep track of this requirement, we equip every crossing vertex~$v_X$ with an \emph{alternation constraint} that~(i) colors its incident edges with colors r(ed), b(lue), p(urple) so that a portion of an edge in~$G$ gets the same color as the corresponding vertex in~$X$, and (ii) specifies which pairs of colors must alternate around~$v$; see \cref{fig:crossing vertex} for an example. 
For a \emph{$K_2$-constraint} there are no purple edges, and red and blue must alternate. For a \emph{$K_3$-constraint} all pairs of colors must alternate. For a \emph{$P_3$-constraint}, red and purple as well as purple and blue must alternate, whereas red and blue must not alternate.  Each component~$X$ of~$\mathcal{C}(A)$ with the coloring described above naturally translates to a constraint for~$v_X$.  For~$X = K_2$, we obtain a~$K_2$-constraint, for~$X = P_3$ we get a~$P_3$-constraint, and for~$X = K_3$ we get a~$K_3$-constraint; see \cref{fig:crossing vertex}.
The auxiliary graph~$H$ with alternation constraints is \emph{feasible}
if it admits a planar embedding that satisfies the alternation constraints of all vertices.
Thus, we have the following.
\begin{lemma}
    \label{lem:reduction-realizability-acp}
    An AT-graph~$A =(G,\Chi)$ with~$\lcc(A) \le 3$ is simply realizable if and only if the corresponding auxiliary graph~$H$ with alternation constraints is feasible.
\end{lemma}

\begin{figure}[t]
    \centering
    \includegraphics[page=1,scale=.92]{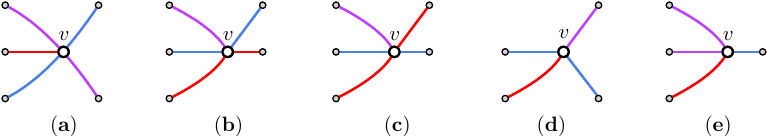}
    \caption{Circular orders of edges incident to a vertex~$v$ satisfying~$(a)$ a~$K_3^{-r}$- and a~$P_3^{-r}$-constraint,~$(b)$ a~$P_3^{-p}$- but not a~$K_3^{-p}$-constraint,~$(c)$ a~$K_3^{-p}$- but not a~$P_3^{-p}$-constraint,$(d)$ a~$P_3^{-(p, r)}$- but not a~$K_3^{-(p, r)}$-constraint ,~$(e)$ a~$P_3^{-(b, r)}$- and a $K_3^{-(b, r)}$-constraint.}
    \label{fig:deg5-constraints}
\end{figure}

 To find such an embedding, we decompose the graph into biconnected components.  It turns out that this may create additional types of alternation
 constraints that stem from the constraints described above, but do not fall into the category of an existing class of constraints.  For the sake of exposition, we introduce these constraints now, even though they will not be part of an instance obtained by 
 the above reduction~from~\SATor.
 
Let~$v$ be a vertex of degree~5 and let~$c$ be a color. For a~$C$-constraint~$(C \in \{K_3, P_3, K_2\})$ as defined above, we define a corresponding~\emph{$C^{-c}$-constraint} of~$v$, which (i) colors the edges incident to~$v$ such that each color occurs at most twice but color~$c$ occurs only once and (ii) requires that in the rotation, it is possible to insert an edge of color~$c$ so that the original~$C$-constraint is satisfied; 
see~\cref{fig:deg5-constraints}. Observe that a~$K_2^{-c}$-constraint is always satisfied and is thus not needed.  Since the colors of a~$K_3$-constraint are entirely symmetric, we may assume without loss of generality that~$c=r$ in this case.  For~$P_3$-constraints, only red and blue are symmetric, i.e., we may assume without loss of generality that either~$c=p$ or~$c=r$.
In particular, the~$K_3^{-r}$-constraint and the~$P_3^{-r}$-constraint both require that purple and blue alternate around~$v$, whereas the position of the red edge is arbitrary; see~\cref{fig:deg5-constraints}\col{$(a)$}. Thus the~$K_3^{-r}$-constraint and the~$P_3^{-r}$-constraint are equivalent. 
For a~$P_3^{-p}$-constraint to be fulfilled, red and blue must not alternate and the purple edge either has to be between the two red edges or between the two blue edges; see~\cref{fig:deg5-constraints}\col{$(b)$}.

Now let~$v$ be a vertex of degree~4 and let~$c, c'$ be two colors.
For a~$C$-constraint, we define a corresponding~\emph{$C^{-c, c'}$-constraint} of~$v$, which (i) colors the edges incident to~$v$ such that the colors distinct from $c$ and $c'$ occur twice but colors~$c$ and $c'$ occur only once if $c \neq c'$, or not at all if~$c = c'$, and (ii) requires that in the rotation, it is possible to insert two edges of color~$c$ and $c'$, respectively, so that the original~$C$-constraint is satisfied.
Since the colors of a~$K_3$-constraint are entirely symmetric, we may assume w.l.o.g. that either~$c = c' = r$ or~$c=r, c' = b$ in this case.  For~$P_3$-constraints, only red and blue are symmetric, we may hence assume without loss of generality that $(c, c') \in \{(r, r), (p,p), (r, p), (r, b)\}$.
Observe that a~$K_2^{-c,c'}$-constraint is always satisfied and is thus not needed.
The same holds for a~$K_3^{-r,b}$-constraint, a~$P_3^{-r,b}$-constraint and a~$P_3^{-r,p}$-constraint; see \cref{fig:deg4-unconstrained}.
Also, note that a $K_3^{-r,r}$-constraint and a~$P_3^{-r, r}$-constraint are both equivalent to a $K_2$-constraint, while a~$P_3^{-p, p}$-constraint requires that red and blue do not alternate around~$v$.

\begin{figure}[t]
    \centering
    \includegraphics[width = 0.9\textwidth]{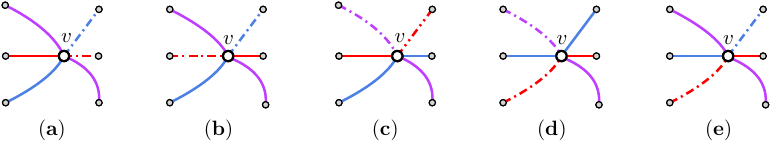}
    \caption{Circular orders of edges around \mbox{a vertex~$v$ allowing to insert two edges of distinct colors} (dashed) so that every color occurs twice \mbox{and 
    a~$(a)$--$(b)$~$K_3$-constraint,~$(c)$--$(e)$ $P_3$-constraint is satisfied.}}
    \label{fig:deg4-unconstrained}
\end{figure}

Finally for a~$C$-constraint, we define a corresponding~\emph{$C^{-(c, c')}$-constraint} of~$v$, which~(i) colors the edges incident to~$v$ such that the colors distinct from~$c$ and $c'$ occur twice but colors~$c$ and $c'$ occur only once if $c \neq c'$, or not at all if $c = c'$, and~(ii) requires that in the rotation, it is possible to insert
an edge of color~$c$ and an edge of color~$c'$ \emph{consecutively}, so that the~$C$-constraint is satisfied; 
see~\cref{fig:deg5-constraints}\col{$(d), (e)$} for examples. 
This type of constraints is motivated as follows.
Let $v$ be a cut vertex in a graph~$G$. 
The \emph{cut components} of $v$ in $G$ are the subgraphs of $G$ induced by $v$ together with the maximal subsets of the vertices of $G$ that are not disconnected by the removal of $v$.
Note that the edges belonging to two different cut components cannot alternate around~$v$ without resulting in a crossing
and observe that a~$K_2^{-(c, c')}$-constraint with~$c \neq c'$ is always satisfied and is thus not needed. 
Also note that a $C^{-(c, c)}$-constraint cannot be satisfied, since every $C$-constraint requires that every color alternates with at least one of the remaining colors.
Since the colors of a~$K_3$-constraint are entirely symmetric, we may assume without loss of generality that~$c=r, c' = b$ in this case.  For~$P_3$-constraints, only red and blue are symmetric, i.e., we may assume without loss of generality that~$(c, c') \in \{(r, p), (r, b)\}$.
In particular, a~$K_3^{-(r, b)}$-constraint and a $P_3^{-(r, b)}$-constraint both require the consecutivity of the two purple edges and are thus equivalent; see~\cref{fig:deg5-constraints}\col{$(e)$}.
For a~$P_3^{-(r, p)}$-constraint to be fulfilled, the two blue edges must not occur consecutively~(see \cref{fig:deg5-constraints}\col{$(d)$}); i.e., the two blue edges have to alternate with the two remaining edges. Thus a~$P_3^{-(p, r)}$-constraint is equivalent to a~$K_2$-constraint.
Table~\ref{tab:constraints} gives an overview over all types of constraints.
By the above discussion we may assume \mbox{that only $K_3$, $P_3$, $K_3^{-r}$, $P_3^{-p}$, $K_2$, $P_3^{-p,p}$ and $K_3^{-(r, b)}$ constraints~occur.}

\setlength{\tabcolsep}{12pt}
 \renewcommand{\arraystretch}{1.7}  
\begin{table}[t]
    \centering
    \begin{tabular}{l | l}
    Constraint & Satisfied if and only if  \\
    \hlineB{4}
        $K_3$ & \tabitem all three colors pairwise alternate \\
        \hline
         $P_3$ & \tabitem red and purple alternate \\[-2ex]
         & \tabitem blue and purple alternate \\[-2ex]
         & \tabitem red and blue do not alternate \\
         \hline
         $K_3^{-r}$, $P_3^{-r}$ 
         & \tabitem blue and purple alternate \\
          \hline
         $P_3^{-p}$ 
         & \tabitem red and blue do not alternate \\
          \hline
         $K_2, K_3^{-r, r}, P_3^{-r, r}$, $P_3^{-(r, p)}$ & \tabitem the two blue edges alternate with the two remaining edges \\
          \hline
         $P_3^{-p, p}$ & \tabitem red and blue do not alternate \\
          \hline
         $K_3^{-(r, b)}$, $P_3^{-(r, b)}$ & \tabitem red and blue are consecutive\\

    \end{tabular}
    \caption{An overview of all non-trivial constraints.}
    \label{tab:constraints}
\end{table}

\noindent The \textsc{Alternation-Constrained Planarity} (\acp) problem has as input a graph~$H$ with alternation constraints and asks whether~$H$ is feasible.
By~\cref{lem:reduction-realizability-acp}, there is a linear-time reduction from \SATR with~$\lcc(A) \le 3$ to~\acp.
Next, we further reduce \acp to \acpB, which is the restriction of \acp to instances for which every connected component of $H$ is 2-connected.

\begin{lemma}
    \label{lem:reduction-to-2-connected}
    There is a linear-time algorithm that either recognizes that an instance~$H$ of \acp is a no-instance or computes an equivalent instance~$H'$
   of~\acpB.
\end{lemma}

\begin{proof}
   Our reduction strategy considers one cut vertex at a time and splits the graph at that vertex into a collection of connected components.  The reduction consists of applying this cut vertex split until all cut vertices are removed or we find out that~$H$ is a no-instance.
Consider an instance~$H$ of~\acp and one of its cut vertices~$v$ with cut components~$H_1, \dots, H_l$.
In the cut components, let every vertex except~$v$ preserve its alternation constraint (if any).
Now the goal is to find out which constraints have to be assigned to the copies of~$v$ in the cut components such that $H$ is a yes-instance if and only if the union of the $H_i$s is a yes-instance.
We denote by~$E(v)$ the edges incident to~$v$ in~$H$ and by $E_i(v)$ the edges incident to~$v$ in~$H_i$, for~$1 \leq i \leq l$.
Without loss of generality assume that~$|E_i(v)| \geq |E_j(v)|$ for~$1 \leq i < j \leq l$. We encode the distribution of edges from~$E(v)$ among the cut components
as a \emph{split-vector}~$(|E_1(v)|, |E_2(v)|, \dots, |E_l(v)|)$.  

If~$v$ has no alternation constraint,~$H$ admits a planar embedding that satisfies all alternation constraints
if and only if each cut component~$H_i$ with~$i = 1, \dots, l$ does and hence all copies of~$v$ remain unconstrained.
Otherwise~$v$ has an alternation constraint~$C$.
Observe that this implies~$|E(v)| \leq 6$ and thus the edges in~$E(v)$ are distributed among at least two and at most six cut components. 
Note that the edges belonging to two different cut components cannot alternate around~$v$ without resulting in a crossing.
Thus~$H$ is a no-instance if~$C \in \{K_3, K_3^{-r}, K_2\}$ and there are two cut components containing a pair of edges of the same color from~$E(v)$, respectively.
If $C = P_3$, the same holds if one cut component contains both purple edges whereas a distinct cut component contains both red or both blue edges. 
In the following, we assume that the above does not apply.
Now we consider cases based on the split-vectors.

\begin{figure}[t]
    \centering
    \includegraphics[width= \textwidth]{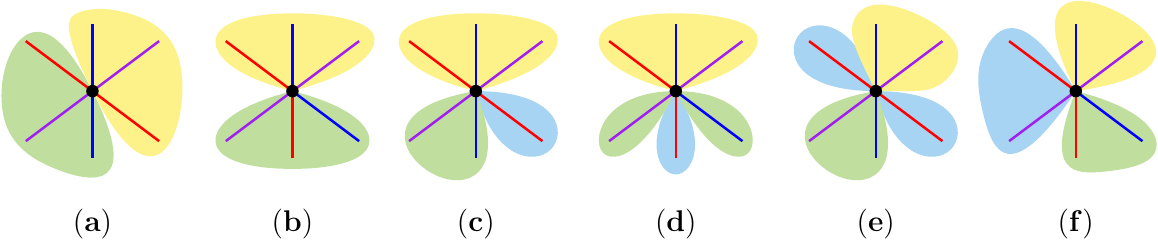}
    \caption{Examples of suitable merges of cut components satisfying a $(a)$~$K_3$-constraint in case $(3, 3)$, 
    $(b)$ $P_3$-constraint in case $(3, 3)$, $(c)$ $K_3$-constraint in case~$(3, 2, 1)$, $(d)$ $P_3$-constraint in case~$(3, 2, 1)$,
    $(e)$ $K_3$-constraint in case~$(2, 2, 2)$,~$(f)$~$P_3$-constraint in case $(2, 2, 2)$.}
    \label{fig:trivial-cases}
\end{figure}

\subparagraph{Case: $|E_1(v)| \leq 3$.}
Note that if $|E_i(v)| = 1$ for every $i > 1$, then we can always arrange the cut components around~$v$ such that $C$ is satisfied.
Now assume that at least two of the cut components contain more than one edge from $E(v)$; i.e., $E_2(v) \geq 2$. We are allowed to $(i)$ mirror the order of the edges in a cut component around~$v$ and to $(ii)$ insert the ordered edges in a cut component consecutively into the 
order of the edges in a different cut component.
It is easy to see that if there are only two cut components $H_1, H_2$,
 we get all cyclic orders of edges around~$v$ where no pair of edges of $H_1$ alternates with a pair of edges from $H_2$; see \cref{fig:trivial-cases}$(a-b)$ for examples.
 Next observe that the same holds if there are more cut components and all but~$H_1, H_2$ contain exactly one edge; see~\cref{fig:trivial-cases}$(c-d)$ for examples.
 It remains to consider the split-vector~$(2, 2, 2)$; i.e.,~$v$ has degree 6 and $C \in \{K_3, P_3\}$. In this case, we get all cyclic orders of edges around~$v$ where no pair of edges of a cut component alternates with a pair of edges from a different cut component. 
 Consider~$C = P_3$. Recall that we assume that no cut component contains both purple edges whereas a distinct cut component contains both red or both blue edges.
 If one component contains both red (or both blue) edges, then the remaining components both contain one blue and one purple edge. Since red has to alternate with purple but must not alternate with blue,~$C$ cannot be satisfied without alternating cut components; i.e., $H$ is a no-instance.
 In all remaining cases with split-vector~$(2, 2, 2)$,
 we can always arrange the cut components around~$v$ such that $C$ is satisfied (see \cref{fig:trivial-cases}\col{$(e)$} for an example), since we assume that there are no two cut components containing a pair of edges of the same color from~$E(v)$, respectively.
 In all positive cases, it suffices to leave each copy of~$v$ unconstrained.
In the following, we consider the remaining split-vectors with $|E_1(v)| \geq 4$.

\subparagraph{Case: $(5, 1)$.}
Let~$C \in \{K_3, P_3\}$ be the constraint of $v$ and let~$c$ be the color of the edge of $E(v)$ in $H_2$.
To merge embeddings of $H_1$ and $H_2$ to a planar embedding of~$H$ such that the $C$-constraint is satisfied, it is necessary that the embedding of~$H_1$ allows to insert an edge of color~$c$ such that the~$C$-constraint is satisfied. Thus it is necessary that the order of edges around~$v$ in~$H_1$ satisfies a~$C^{-c}$-constraint.
Note that if the $C^{-c}$-constraint is satisfied, it is guaranteed that the embeddings of $H_1$ and $H_2$ can be merged such that the original~$C$-constraint is satisfied.
Thus it is necessary and sufficient
to equip the copy of~$v$ in~$H_1$ with a $C^{-c}$-constraint whereas the copy of~$v$ in $H_2$ remains unconstrained.

\subparagraph{Case: $(4, 2)$.}
Let~$C \in \{K_3, P_3\}$ be the constraint of $v$ and let~$c, c'$ be the colors of the edges of $E(v)$ in $H_2$.
To merge embeddings of $H_1$ and $H_2$ to a planar embedding of~$H$ such that the $C$-constraint is satisfied, it is necessary that the embedding of~$H_1$ allows to insert two edges of color~$c, c'$ consecutively, such that the~$C$-constraint is satisfied. Thus, it is necessary that the order of edges around~$v$ in~$H_1$ satisfies a~$C^{-(c, c')}$-constraint.
Note that if the $C^{-(c, c')}$-constraint is satisfied, the embeddings of $H_1$ and $H_2$ can always be merged such that the original $C$-constraint is satisfied.
Thus it is necessary and sufficient
to equip the copy of~$v$ in $H_1$ with a $C^{-(c, c')}$-constraint whereas the copy of~$v$ in~$H_2$ remains unconstrained.

\subparagraph{Case: $(4, 1, 1)$.}
Let~$C \in \{K_3, P_3\}$ be the constraint of $v$ and let~$c, c'$ be the colors of the edges of $E(v)$ in $H_2$ and $H_3$, respectively.
To merge embeddings of $H_1$ and $H_2$ to a planar embedding of~$H$ such that the $C$-constraint is satisfied, it is necessary that the embedding of~$H_1$ allows to insert two edges of color~$c, c'$, such that the~$C$-constraint is satisfied. Thus it is necessary that the order of edges around~$v$ in~$H_1$ satisfies a~$C^{-c, c'}$-constraint.
Note that if the $C^{-c, c'}$-constraint is satisfied, it is guaranteed that the embeddings of $H_1$ and $H_2$ can be merged such that the original $C$-constraint is satisfied.
Thus it is necessary and sufficient
to equip the copy of~$v$ in $H_1$ with a $C^{-c, c'}$-constraint whereas the copies of~$v$ in $H_2$ and $H_3$ remain unconstrained.

\subparagraph{Case: $(4, 1)$.}
Let~$C^{-c}$ with~$C \in \{K_3, P_3\}$ be the constraint of $v$ and let~$c'$ be the color of the edge in $E_2(v)$.
To merge embeddings of $H_1$ and $H_2$ to a planar embedding of~$H$ that satisfies the $C^{-c}$-constraint, it is necessary that the embedding of~$H_1$ allows to insert an edge of color~$c$, such that the~$C^{-c, c'}$-constraint is satisfied. Thus it is necessary that the order of edges around~$v$ in~$H_1$ satisfies a~$C^{-c, c'}$-constraint.
Note that if the $C^{-c, c'}$-constraint is satisfied, the embeddings of $H_1$ and $H_2$ can always be merged such that the original~$C^{-c}$-constraint is satisfied.
Thus it is necessary and sufficient
to equip the copy of~$v$ in $H_1$ with a~$C^{-c, c'}$-constraint whereas the copy of~$v$ in $H_2$ remains unconstrained.

\smallskip
\noindent
We may assume that after a linear-time preprocessing every edge in~$H$ is labeled with~the block it belongs to.
Then for a cut vertex~$v$ a split as described above takes~$O(\deg(v))$-time.
When no cut vertex is left, we return the~obtained constrained~blocks.
\end{proof}

\subparagraph{Algorithm for the Embedding Problem.}
We define a more general problem \textsc{General Alternation-Constrained Planarity} (\acpG) whose input is a  graph~$H$ where vertices of degree 4, 5, or 6 may be equipped with an alternation constraint or with a (synchronized) PQ-tree ({\em but not both}).
The question is whether~$H$ admits a planar embedding such that all alternation constraints are satisfied (i.e., $H$ is feasible)
and the order of edges around a vertex with a PQ-tree~$B$ is compatible with~$B$.
The \acpGB problem is the restriction of \acpG to input graphs whose connected components are 2-connected.
Clearly, every instance of (\textsc{2-connected})~\acp is an instance of (\textsc{2-connected})~\acpG. 
For our purpose, however, it will turn out that PQ-tree constraints are easier to handle.
Thus, given an instance of~\acp we aim to construct an equivalent instance of~\acpG, where as many alternation constraints as possible are replaced by PQ-trees.
In particular, alternation constraints of degree-$4$ vertices can be replaced by the PQ-trees shown in \cref{fig:pq-reduction}. 

\begin{lemma}
\label{lem:reductionGACP}
 Given an instance of \textsc{2-connected} \acpG containing a \mbox{degree-$4$} vertex~$v$ with alternation constraint~$C$, 
  we obtain an equivalent instance by removing $C$ and equipping~$v$ with a suitable (synchronized) PQ-tree.
\end{lemma}

\begin{proof}
    Let $C$ be the constraint of~$v$.
    We distinguish cases based on~$C$.

 \subparagraph{Case:~$C$ is a~$K_2$-constraint.}
 The constraint is fulfilled if and only if red and blue alternate, which can be 
 represented by a single Q-node as shown 
 in \cref{fig:pq-reduction}\col{$(a)$}.

   \subparagraph{Case:~$C$ is a~$P_3^{-p, p}$-constraint.}
 The constraint is fulfilled if and only if red and blue do not alternate. Hence it is necessary and sufficient that the two red edges appear consecutively, which can be enforced by the PQ-tree 
 in~\cref{fig:pq-reduction}\col{$(b)$}.

\subparagraph{Case:~$C$ is a~$K_3^{-(r,b)}$-constraint.}
The constraint is fulfilled if and only if a red and a blue edge can be inserted consecutively, such that all colors pairwise alternate. 
This is the case if and only if the two purple edges appear consecutively, which can be enforced by
 the PQ-tree in \cref{fig:pq-reduction}\col{$(c)$}.
\end{proof}

Hence we may assume from now on that no vertex with an alternation constraint in $H$ has degree less or equal to~4; i.e., all these vertices have degree~5 or~6.

\begin{figure}[t]
\centering
\includegraphics[page=13]{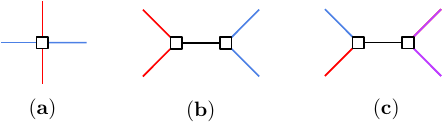}
\caption{The PQ-trees representing alternation constraints of degree-4 vertices.
(a) $K_2$-constraint,
(b) $P_3^{-p, p}$-constraint, and
(c) $K_3^{-(r,b)}$-constraint.}
    \label{fig:pq-reduction}
\end{figure}

Let~$v$ be a vertex with alternation constraints. We call two edges~$e, f$
incident to~$v$ a \emph{consecutive edge pair}, if they are consecutive (around~$v$) in {\em every} planar embedding of~$H$ that satisfies all constraints.
We show that, in most cases, an alternation constraint at a vertex incident to a consecutive edge pair can be replaced by a PQ-tree.

\begin{figure}[t]
    \centering
    \includegraphics[page=4]{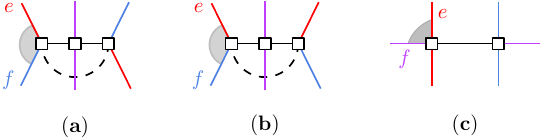}
    \caption{The (synchronized) PQ-trees for degree-6 vertices with consecutive edge pair and~$(a)$~$K_3$-constraint,~$(b), (c)$ $P_3$-constraint.}
      \label{fig:consec6}
   \end{figure}

     \begin{figure}[t]
    \centering
    \includegraphics[page=9]{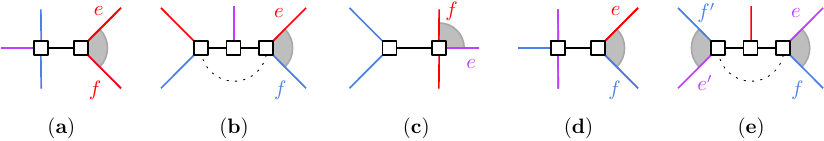}
    \caption{The (synchronized) PQ-trees for degree-5 vertices with consecutive edge pair and~$(a -~c)$~$P_3^{-p}$-constraint, $(d)$~$K_3^{-r}$-constraint.}
    \label{fig:consec5P3}
    \end{figure}

\begin{lemma}
    \label{lem:consecConstraints}
    Let~$H$ be an instance of \textsc{2-connected}~\acpG and let $v$ be
    a vertex in~$H$ with alternation constraint~$C$  incident to a consecutive edge pair.
    Then~$H$ is either a no-instance, or we obtain an equivalent instance~$H'$ by removing~$C$ and equipping~$v$ with a suitable (synchronized) PQ-tree if
    \begin{enumerate}[(i)]
        \item $C \neq K_3^{-r}$ or
        \item $C = K_3^{-r}$ and one of the consecutive edges is red or
        \item $C = K_3^{-r}$ and $v$ is incident to two distinct consecutive edge-pairs.
    \end{enumerate}
\end{lemma}

\begin{proof}
Let $e, f$ be a consecutive edge pair. 
Note that if~$\deg(v) = 6$, we may assume that~$e, f$ have distinct colors, since otherwise $H$ is a no-instance as every color has to alternate with at least one other color in order to satisfy a $K_3$- or $P_3$-constraint.

\subparagraph{Case:~$C$ is a~$K_3$-constraint.}
Without loss of generality assume that~$e$ and~$f$ are red and blue, respectively.
Then the orders of edges around~$v$ in which $e$, $f$ appear consecutively and the alternation constraint is satisfied, can be represented by the synchronized PQ-tree shown in \mbox{\cref{fig:consec6}\col{$(a)$}}.

\subparagraph{Case:~$C$ is a~$P_3$-constraint.}
First consider the case that~$e$ and~$f$ are red and blue. Then the synchronized PQ-tree shown in \mbox{\cref{fig:consec6}\col{$(b)$}} represents all allowed orders around~$v$.
Otherwise, if~$e$ is purple we may assume that~$f$ is red and label~$v$ with the PQ-tree shown in~\mbox{\cref{fig:consec6}\col{$(c)$}}.

\subparagraph{Case:~$C$ is a~$P_3^{-p}$-constraint.} 
If~$e$ and~$f$ have the same color, say red, we label~$v$ with the PQ-tree shown in \mbox{\cref{fig:consec5P3}\col{$(a)$}}.
If~$e$ and~$f$ are red and blue, respectively, we label~$v$ with the synchronized PQ-tree shown in \mbox{\cref{fig:consec5P3}\col{$(b)$}}.
Finally, if~$e$ is purple, then~$f$ has a different color, say red. Then we label~$v$ with the PQ-tree shown in \mbox{\cref{fig:consec5P3}\col{$(c)$}}.

\subparagraph{Case:~$C$ is a~$K_3^{-r}$-constraint and one of the consecutive edges is red.} 
In this case~$e$ and~$f$ must be of distinct colors and we may assume without loss of generality that $e$ is red and $f$ is blue. Then we label~$v$ with the PQ-tree shown in \mbox{\cref{fig:consec5P3}\col{$(d)$}}.

\subparagraph{Case:~$C$ is a~$K_3^{-r}$-constraint and $v$ is incident to two distinct consecutive edge pairs.} 
If the red edge is contained in a consecutive pair, we are in the previous case. Hence assume this is not the case.
Note that~$H$ is a no-instance if a consecutive edge pair contains two edges of the same color.
Hence we may assume that both consecutive edge pairs~$\{e,f\}$ and~$\{e',f'\}$ contain a blue and a purple edge. If the two consecutive pairs are disjoint, then we label~$v$ with the synchronized PQ-tree shown in \mbox{\cref{fig:consec5P3}\col{$(e)$}}.  Otherwise, we may assume that~$f=f'$.  But then the two pairs enforce that~$\{e,e',f\}$ is consecutive, and consequently we obtain that the remaining two edges, one of which is red, also form a consecutive pair.  We can thus apply the previous case.
\end{proof}

The overall strategy of the remaining section consists of three steps.
In Step 1 we identify consecutive edge pairs in~$H$ that allow us to replace alternation constraints by PQ-trees with the help of the SPQR-tree of~$H$. Then we show how to compute an equivalent instance that does not contain~$P_3$- or $K_3^{-r}$-constraints. Finally, we end up with an instance whose alternation constraints are all $K_3$-constraints and every vertex with such a constraint appears in the skeletons of exactly two $P$-nodes and one $S$-node in the SPQR-tree.
In Step~2, we handle such constraints by considering them on a more global scale. We show that they form cyclic structures, where either the constraints cannot be satisfied or can be dropped and satisfied irrespective of the remaining solution.  Eventually, we arrive at an instance with only (synchronized) PQ-trees as constraints, which we solve with standard techniques in Step~3.

For the rest of this section let $H$ be an instance of \acpGB and let~$T$ be the SPQR-tree of~$H$. We begin with Step~1 and identify consecutive edge pairs.
Let $\mu$ be a node of~$T$ and let~$v$ be a vertex of its skeleton incident to the virtual edges~$e_1,\dots,e_k$.
Then the \emph{distribution vector} $(d_1, \dots, d_k)$ of~$v$, with $d_i \geq d_{i+1}$ for every~$1 \leq~i <~k$, contains for each virtual edge~$e_i$ the number~$d_i$ of edges from $E(v)$ contained~in~$e_i$.

\begin{lemma}
\label{lem:R-nodes}
    Let $H$ be an instance of \textsc{$2$-connected}~\acpG and let~$T$ be the SPQR-tree of~$H$.
    A vertex~$v$ with alternation constraint~$C$ in $H$ is incident to a consecutive edge pair if~$v$ appears in a skeleton of~$T$ that
    \begin{enumerate}[(i)]
        \item  has a virtual edge that contains exactly two edges from~$E(v)$, 
        \item has a virtual edge that contains all but two edges from~$E(v)$, or
        \item is an $R$-node.
    \end{enumerate}
    Moreover, a vertex~$v$ of degree~$5$ is incident to two consecutive edge pairs if $v$ appears in a skeleton of~$T$ that
    \begin{enumerate}[(i)]
    \setcounter{enumi}{3}
        \item has two virtual edges, each of which contains exactly two edges from~$E(v)$, or
        \item is an $R$-node that has at least four virtual edges incident to~$v$.
    \end{enumerate}
\end{lemma}

\begin{proof}
    It is easy to see that in cases $(i)$ and $(ii)$ there is a consecutive edge pair.
    For~$(iii)$ consider a vertex~$v$ in the skeleton of an $R$-node and assume that neither~$(i)$ nor $(ii)$ applies.
    Recall that~$v$ has degree~5 or 6 and
that the vertices in the skeleton of an~$R$-node have degree at least~3; i.e., the edges of~$E(v)$ are distributed among at least three virtual edges.
   Thus it remains to consider the distribution vectors $(1, 1, 1, 1, 1, 1), (1, 1, 1, 1, 1), (3, 1, 1, 1)$.
  Recall that the order of the virtual edges around a vertex in the skeleton of an $R$-node
is fixed up to reversal. 
In all three cases, there are two edges~$f, g$ appearing in two virtual edges~$e_f, e_g$, respectively, such that 
$e_f, e_g$ both contain exactly one edge from $E(v)$ and appear consecutively in every planar embedding of the skeleton of the $R$-node.
Thus, $\{f, g\}$ is a consecutive edge pair.

It is easy to see that in case~$(iv)$ there are two consecutive edge pairs.
For~$(v)$ consider a vertex~$v$ in the skeleton of an $R$-node~$\mu$ and assume that~$(iv)$ does not apply.
Thus it remains to consider the distribution vectors~$(2, 1, 1, 1)$ and~$(1, 1, 1, 1, 1)$. 
Recall  that the order of the virtual edges around a vertex in the skeleton of an $R$-node is fixed up to reversal. In both cases there are three virtual edges~$e, e', e''$, each of which contains exactly one edge from~$E(v)$ that appear consecutively around~$v$ with~$e'$ between $e$ and $e''$ in every planar embedding of the skeleton of~$v$. 
Let $f, g, h$ be the corresponding edges in~$E(v)$ such that~$f \in e, g \in e', h\in e''$. Then $\{f, g\}$ and $\{g, h\}$ are two consecutive edge pairs.
\end{proof}

Since we immediately replace alternation constraints by PQ-trees whenever we find consecutive edge pairs, we assume from now on that no vertex with alternation constraint satisfies one of the conditions of \cref{lem:R-nodes}.  We will show later that this implies that every vertex that has an alternation constraint shows up in the skeleton of a $P$-node.  Thus, in the following we focus on vertices~$v$ with alternation constraints that show up in the skeleton of a P-node and show how to get rid of their alternation constraints.  By the assumption that none of these vertices satisfies the conditions of \Cref{lem:R-nodes}, we can exclude distribution vectors that contain two 2s as well as distribution vectors that contain a 2 and that sum to 6. It thus remains to consider vertices that appear in $P$-nodes with the distribution vectors~$(1, 1, 1, 1, 1, 1), (1, 1, 1, 1, 1), (3, 1, 1, 1), (3, 1, 1)$ and~$(2, 1, 1, 1)$.
For each of these cases we have to distinguish three subcases, as the second pole~$u$ of the $P$-node is either unconstrained, has a PQ-constraint, or also has an alternation-constraint. In the following we show that in most cases we can compute an equivalent instance with fewer alternation constraints in constant time and thus may assume that these cases do not occur; see~\cref{fig:table} for an overview.
We start by showing that in the case of a $(1, 1, 1, 1, 1, 1)$ or a $(1, 1, 1, 1, 1)$ distribution vector we can get rid of the alternation constraint~$C$ of~$v$ as it is either always possible to reorder the children of~$\mu$ according to~$C$ within a realization of~$H$ without~$C$, or~$H$ without~$C$ (and thus~$H$) is not realizable.

\begin{figure}
    \centering
    \includegraphics[]{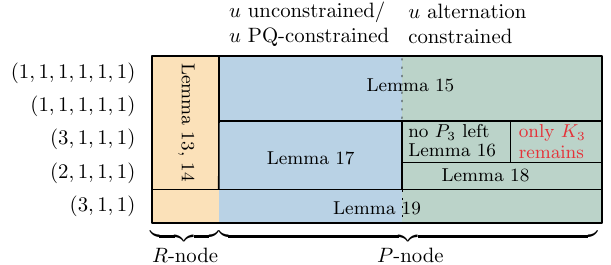}
    \caption{Overview of which cases are handled by which lemma.}
    \label{fig:table}
\end{figure}

\begin{lemma}
  \label{lem:(11111)}
 Let~$v$ be a vertex in~$H$ with alternation constraint~$C$ that appears in a $P$-node~$\mu$ in $T$ with distribution vector $(1, 1, 1, 1, 1, 1)$ or $(1, 1, 1, 1, 1)$. There is a constant-time algorithm that either recognizes that~$H$ is a no-instance, or
  computes an equivalent instance~$H'$ with fewer alternation constraints by removing $C$ and possibly other constraints.
\end{lemma}

\begin{proof}
  Let $u$ be the pole of $\mu$ distinct from~$v$.
   If~$u$ is unconstrained let $H'$ be the instance we obtain by removing $C$ from $H$.
   Since every solution for $H'$ remains a valid solution if we reorder the children of~$\mu$ according to~$C$, $H$ and $H'$ are equivalent.
  Now assume that~$u$ also has an alternation constraint~$C'$. 
  Recall that $u$ and $v$ have either degree 5 or 6. If~$\deg(u) \neq \deg(v)$, say $\deg(u) = 6, \deg(v) = 5$, then 
  there is a virtual edge in the skeleton of~$\mu$ that contains exactly two edges $e, f$ from $E(u)$; a contradiction 
to the assumption that no degree-6 vertex with alternation constraint is incident to a consecutive edge pair. Thus~$\deg(u) = \deg(v)$ and we check whether
one of the constantly many orderings of the virtual edges satisfies both~$C$ and~$C'$. 
In the negative case, $H$ is a no-instance;
in the positive case, let $H'$ be the instance we obtain by removing both ~$C$ and~$C'$ from $H$.
Clearly, every solution of $H$ is a solution for~$H'$.
Note that every solution for~$H'$ remains valid if we reorder the children of~$\mu$ according to one of the orders that satisfy both $C$ and~$C'$. Thus $H$ and~$H'$ are~equivalent.

It remains to consider the case that~$u$ has a PQ-tree~$B$.
We may assume without loss of generality that $B$ contains the consecutivities induced by the virtual edges in the skeleton of~$\mu$.
If~$\deg(u) = \deg(v)$, we check whether
one of the constantly many orderings of the virtual edges satisfies~$C$ and is compatible with~$B$.
In this case, similarly to the discussion above, we get an equivalent instance~$H'$ by removing~$C$.
Otherwise,~$\deg(v) = 5$ and $\deg(u) = 6$, i.e.,~$v$ has distribution vector~$(1, 1, 1, 1, 1)$ in~$\mu$, whereas~$u$ has distribution vector~$(2, 1, 1, 1, 1)$ in~$\mu$.
Let~$e, f$ be the two edges from $E(u)$ contained in the same virtual edge~$e_1$.
Then, similar to the case above, we check if one of the constantly many orderings of the virtual edges satisfies~$C$ and is compatible with~$B$.
In the negative case, $H$ is a no-instance.
In the positive case we get an equivalent instance~$H'$ by removing~$C$ from $H$, as
every solution for~$H'$ remains valid if we reorder the children of~$\mu$ according to one of the orders that satisfies both~$C$ and~$B$ and, if necessary, flip the embedding of the subgraph represented~by~$e_1$. 
\end{proof}

Similar techniques allow us to get rid of~$P_3$-constraints.

\begin{lemma}
\label{lem:(3111)P3}
  Let~$v$ be a vertex in~$H$ with $P_3$-constraint $C$ that appears in a $P$-node~$\mu$ in~$T$ with distribution vector $(3, 1, 1, 1)$. 
  Then there is a constant-time algorithm that either recognizes that~$H$ is a no-instance, or
  computes an equivalent instance $H'$ with fewer alternation constraints by replacing~$C$ by a suitable (synchronized) PQ-tree.
\end{lemma}

\begin{proof}
    Note that $H$ is a no-instance, if~$e_1$ contains either both or none of the purple edges from $E(v)$.
    Hence assume that~$e_1$ contains exactly one purple edge~$f$ from $E(v)$.
    If $e_1$ contains a red edge~$g$ and a blue edge~$h$ from $E(v)$, then 
    no planar embedding of~$H$ with~$f$ between~$g$ and $h$ around~$v$ satisfies~$C$, since then purple cannot alternate with both blue and red such that red and blue do not alternate.
    Thus $g, h$ are a consecutive edge pair and~$C$ can be replaced by a suitable (synchronized) PQ-tree by \mbox{\cref{lem:consecConstraints}}.
    Otherwise, assume without loss of generality that $e_1$ contains one purple and two red edges from~$E(v)$. In this case, in any planar embedding of~$H$ that satisfies~$C$,
    the purple edge~$f$ has to appear between the two red edges.
    Thus~$f$ forms a consecutive edge pair with each of the red edges and thus again by \mbox{\cref{lem:consecConstraints}},~$C$ can be replaced by a suitable (synchronized) PQ-tree.
\end{proof}

Next we show that if only one pole of a $P$-node has an alternation constraint, then we can also get rid of this alternation constraint in case of a $(3, 1, 1, 1)$ or a $(2, 1, 1, 1)$ distribution.

    \begin{figure}[t]
        \centering
        \includegraphics{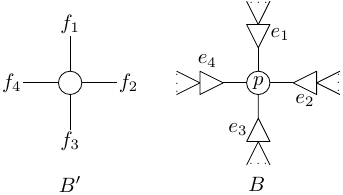}
        \caption{Illustrations of the PQ-trees used in the proof of \cref{lem:(3111)unconstrainedPQ}. Triangles represent subtrees of~$B$.}
        \label{fig:trivialPQtree}
    \end{figure}

\begin{lemma}
\label{lem:(3111)unconstrainedPQ}
  Let~$v$ be a vertex in~$H$ with alternation constraint $C \in \{K_3, K_3^{-r}\}$ that appears in a $P$-node~$\mu$ in $T$ with distribution vector $(3, 1, 1, 1)$ or $(2, 1, 1, 1)$. 
  Let $u$ be the pole of~$\mu$ distinct from~$v$. If~$u$ is unconstrained or has a (synchronized) PQ-constraint, 
  there is a constant-time algorithm that either recognizes that~$H$ is a no-instance,
  or computes an equivalent instance~$H'$ with fewer alternation constraints by either removing~$C$ or replacing~$C$ with a suitable (synchronized) PQ-tree.
\end{lemma}

\begin{proof}
     Observe that in case of a $(3, 1, 1, 1)$ distribution vector, $C = K_3$ and the three edges from $E(v)$ in $e_1$ must have pairwise distinct colors; otherwise~$H$ is a no-instance.
    If the distribution vector of~$v$ is $(2, 1, 1, 1)$, $C = K_3^{-r}$ and the two edges from $E(v)$ in $e_1$ must be purple and blue; otherwise $H$ is a no-instance or the alternation constraint can be replaced by a PQ-tree by~\mbox{\cref{lem:consecConstraints}}.
     If~$u$ is unconstrained, let $H'$ be the instance we obtain from $H$ by removing~$C$.
    Clearly, every solution of $H$ is also a solution for $H'$.
    Note that every solution for~$H'$ remains a valid solution if we reorder the children of~$\mu$ according to~$C$. Thus,~$H$ is a yes-instance if and only if $H'$ is a yes-instance.

    Now assume that $u$ has a (synchronized) PQ-constraint and let~$B$ be the (synchronized) PQ-tree of~$u$. Note that we may assume without loss of generality that $B$ contains the consecutivities induced by the virtual edges in the skeleton of~$\mu$.
    Let~$e_1, e_2, e_3, e_4$ be the four virtual edges incident to~$u$ in the skeleton of~$\mu$ and 
    let~$f_1, f_2, f_3, f_4$ be four edges from~$E(u)$ such that for every $1 \leq i \leq 4$,~$f_i$ is contained in $e_i$.
    We set~$F = \{f_1, f_2, f_3, f_4\}$ and denote the restriction of~$B$ to~$F$ by~$B'$.
    Observe that~$B'$ has four leaves and thus contains at most two non-leaf nodes. Moreover, note that in case $B'$ is synchronized, there is an equivalent PQ-tree containing exactly one non-leaf node of type~$Q$, thus we may assume that~$B'$ is not a synchronized PQ-tree.
    If $B'$ is trivial, i.e., it contains exactly one non-leaf node that is of type~$P$, then we let $H'$ be the instance we obtain from~$H$ by removing~$C$.
    Clearly, every solution of~$H$ is also a solution of~$H'$.
    Conversely, note that the structure of~$B'$ implies that in $B$ there is a $P$-node~$p$ of degree~4 such that for every virtual edge $e_i$ in the skeleton of~$\mu$, there is an edge in~$B$ incident to~$p$ whose removal yields two subtrees $B_1$, $B_2$ such that the leaf-set of~$B_1$ corresponds exactly to the edges from $E(u)$ in~$e_i$; see~\mbox{\cref{fig:trivialPQtree}}.
    Hence the embeddings of the subgraphs represented by the virtual edges in the skeleton of~$\mu$ are completely independent.
    In other words, every solution of~$H'$ remains a valid solution after reordering the children of~$\mu$ according to~$C$, since~$B$ is compatible with each of these orders.
    Thus $H$ is a yes-instance if and only if~$H'$ is a yes-instance.

    Otherwise, there is an edge pair~$f_i, f_j$ in~$F$ that is consecutive in every order represented by~$B'$.
    Note that also~$F \setminus \{f_i, f_j\}$ is consecutive in every order represented by~$B'$, and that this is also a consecutive a pair of edges, since $B'$ has only four leaves.
    Hence we may assume without loss of generality that $i, j \neq 1$. 
    In particular, we get that two of the virtual edges containing exactly one edge from $E(v)$ have to appear consecutively in any planar embedding of~$H$ that satisfies all constraints. 
    This yields a consecutive edge pair around~$v$. 
    Note that in case $C= K_3^{-r}$ the two edges of~$E(v)$ in $e_1$ are a second consecutive edge pair.
    Thus, independent of its type,  $C$ can be replaced by a suitable (synchronized) PQ-tree by \mbox{\cref{lem:consecConstraints}}.
\end{proof}

\begin{figure}
    \centering
    \includegraphics{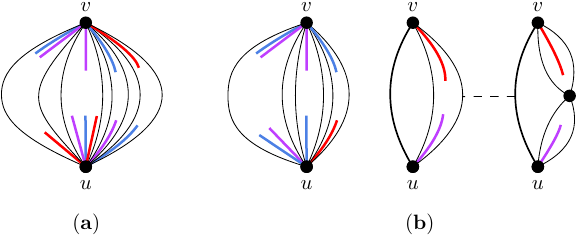}
    \caption{Illustrations of operations used in the proof of~\cref{lem:(2111)}.}
    \label{fig:P(2111)}
\end{figure}

\begin{figure}[t]
    \centering
    \includegraphics{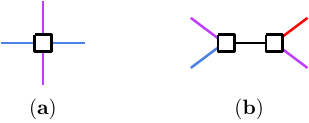}
    \caption{The PQ-trees used in the proof of~\cref{lem:(2111)}.}
    \label{fig:PQnew}
\end{figure}

In case of a $(2, 1, 1, 1)$ distribution in a $P$-node~$\mu$ with two alternation constrained poles, we get rid of both alternation constraints in constant time.

\begin{lemma}
\label{lem:(2111)}
Let~$v$ be a vertex in~$H$ with $K_3^{-r}$-constraint $C$ that appears in a $P$-node~$\mu$ in~$T$ with distribution vector~$(2, 1, 1, 1)$. 
  Let $u$ be the pole of~$\mu$ distinct from~$v$. If~$u$ also has an alternation constraint~$C'$, then 
  we can compute an equivalent instance~$H'$ with fewer alternation constraints in constant time.
\end{lemma}

\begin{proof}
    Since~$u$ has an alternation constraint, it has either degree~5 or 6.
    Let~$e_1^v$ be the virtual edge in~$\mu$ that contains more than one edge from~$E(v)$ and let $e_1^u$ be the virtual edge in~$\mu$ that contains more than one edge from~$E(u)$.
    First assume $e_1^v \neq e_1^u$. Let~$H'$ be the instance we obtain from~$H$ by removing~$C$.
    Clearly every solution of~$H$ is also a solution for~$H'$. Note that every solution for~$H'$ remains valid if we reorder the children of~$\mu$ according to~$C'$ and, if necessary, flip the embedding of~$e_1^v$ in order to satisfy~$C$; see~\cref{fig:P(2111)}$(a)$ for an illustration.
    
    If~$e_1^v = e_1^u$, let~$e_r$ denote the virtual edge in~$\mu$ that contains the red edge in $E(v)$. 
    Let~$G_1$ be the graph represented by $e_r$ in~$\mu$ plus a $uv$-edge. Let $G_2$ be the graph represented by the virtual edges distinct from~$e_r$ in~$\mu$. Note that both~$G_1$ and $G_2$ are biconnected.
    Let~$H'$ be the instance we obtain from~$H$ by $(i)$ replacing the biconnected component containing~$v$ by $G_1$ and $G_2$ such that every vertex except for the copies of~$u,v$ keeps its constraint (if any) and~$(ii)$ equipping the copy
    of~$v$ in~$G_2$ with the PQ-tree shown in~\cref{fig:PQnew}$(a)$ and~$(iii)$
    equipping the copy
    of~$u$ in~$G_2$ with a suitable constraint depending on its degree and the distribution of its edges.
    The constraint for $u$ depends on the degree of~$u$.
    First assume $\deg(u) = 5$. Then $u$ has a $K_3^{-r}$-constraint in~$H$ and the two edges of~$E(u)$ in $e_1^v$ are blue and purple. 
     If the red edge in~$E(u)$ is contained in $G_1$, then we label the copy of~$u$ in~$G_2$ with the  PQ-tree shown in~\cref{fig:PQnew}$(a)$. Otherwise we may assume without loss of generality that the copy of~$u$ in $G_2$ is incident to two purple edges and thus we label it with the PQ-tree shown in~\cref{fig:PQnew}$(b)$.
     It remains to consider the case $\deg(u) = 6$.
     By~\cref{lem:(3111)P3} we may assume that~$u$ has a $K_3$-constraint in~$H$ and
     that the edges of~$E(u)$ in $e_1^v$ have pairwise distinct colors. As the edge colors are completely symmetric in this case, we may assume that the edge of~$E(u)$ contained in~$e_r$ is red and we equip the copy of~$u$ in~$G_2$ with a $K_3^{-r}$-constraint.
     
    From every solution of~$H$ we obtain a solution of~$H'$ by copying the embedding of  the connected component of~$v$, removing unnecessary vertices and edges from the first copy and replacing unnecessary vertices and edges by a single $uv$-edge in the second copy.
    Conversely, assume that we have a solution for~$H'$. Since the position of the red edge around~$v$ is arbitrary, it is always possible to
    reinsert the graph represented by~$e_r$ into an embedding of the remaining graph, such that both alternation constraints are fulfilled. 

    Note that we obtain the SPQR-tree~$T'$ of~$H'$
    from $T$ in constant-time as follows.
    Let~$\nu$ be the node adjacent to~$\mu$ that contains the twin edge of~$e_r$. Note that $\nu$ is an $S$-node since~$v$ has degree~$2$ in the skeleton of~$\nu$.
    We remove~$e_r$ from~$\mu$ and add a new $Q$-node adjacent to~$\nu$; see~\cref{fig:P(2111)}$(b)$ for an illustration. 
\end{proof}

Finally we show that every vertex that appears with a distribution vector~$(3, 1, 1)$ in the SPQR-tree, also appears in a $P$-node with distribution vector~$(2, 1, 1, 1)$. By~\cref{lem:(2111),lem:(3111)unconstrainedPQ}, we may thus assume  that no vertex with alternation constraint appears with distribution vector~$(3, 1, 1)$ in~$T$.

\begin{lemma}
    \label{lem:311}
    Let~$v$ be a vertex with alternation constraint~$C$ that appears in a node~$\mu$ in~$T$ with distribution vector~$(3, 1, 1)$.
    Then there is a $P$-node~$\nu$ in~$T$ that contains~$v$ with distribution vector~$(2,1 ,1 ,1)$.
\end{lemma}

\begin{proof}
    The distribution vector~$(3,1, 1)$ implies that there is a node~$\nu$ in $T$ containing~$v$, such that the three edges of~$E(v)$ that appear in the same virtual edge in~$\mu$, are distributed among at least two virtual edges in~$\nu$. Moreover, 
there is a distinct virtual edge in the skeleton of~$\nu$ that contains the two remaining edges from~$E(v)$. Hence the distribution vector of~$v$ in~$\nu$ is either $(2, 2, 1)$ or $(2, 1, 1, 1)$.
In the first case we immediately get two consecutive edge pairs by~\cref{lem:R-nodes}$(iv)$, thus we assume  the latter case.
Note that~$\nu$ cannot be an $S$-node, since~$v$ has degree~$4$
in the skeleton of~$\nu$.
If~$\nu$ is an $R$-node, we 
have two consecutive edge pairs by~\cref{lem:R-nodes}~$(v)$, hence~$v$ appears in a $P$-node with distribution vector~$(2, 1, 1, 1)$.
\end{proof}

Next we show that by applying~\cref{lem:consecConstraints,lem:R-nodes,lem:(11111),lem:(2111),lem:(3111)P3,lem:(3111)unconstrainedPQ} we can assume from now on that only $K_3$-constraints and PQ-constraints occur and that every vertex with alternation constraint appears in a $P$-node in~$T$ with distribution vector~$(3,1, 1, 1)$.

\begin{lemma}
    \label{lem:only3111left}
    Let~$H$ be an instance of \acpGB. There is a linear-time algorithm that either recognizes that~$H$ is a no-instance, or computes an equivalent instance~$H'$ such that every alternation constraint is a $K_3$-constraint and every vertex with $K_3$-constraint appears in the skeleton of a $P$-node with distribution vector~$(3, 1, 1, 1)$.
\end{lemma}

\begin{proof}
    Let~$v$ be a vertex in~$H$ with alternation constraint~$C$. First we show that~$v$ appears in a $P$-node in~$T$.
    Assume that~$v$ appears in an $R$-node~$\mu$. Then by~\cref{lem:consecConstraints,lem:R-nodes} and since every vertex in the skeleton of an $R$-node has degree at least~3, we may assume that the distribution vector
    of~$v$ in $\mu$ is~$(3, 1, 1)$. Now by~\cref{lem:311},~$v$ appears in a $P$-node.
Next assume that~$v$ appears in an $S$-node $\nu$ in~$T$.
Since the vertices in the skeleton of an~$S$-node have degree~2 and~$v$ has degree~5 or 6, 
there is a node~$\mu$ adjacent to~$\nu$ in $T$ also containing~$v$. Note that~$\mu$ is either a $P$-node or an $R$-node, as there are no two adjacent $S$-nodes in an SPQR-tree.
If~$C \neq K_3^{-r}$,~$\mu$ is a $P$-node by~\mbox{\cref{lem:R-nodes}}$(iii)$. Otherwise, if $C = K_3^{-r}$ then $\mu$ is either a $P$-node or an $R$-node. 
In the latter case it follows by the above discussion that~$v$ appears in a $P$-node.
Hence every vertex with alternation constraint appears in a $P$-node in~$T$.

Recall that the vertices in the skeleton of a~$P$-node have degree at least~3; i.e., the edges of~$E(v)$ are distributed among at least three virtual edges.
The only possible distributions without consecutive edge pairs that can be handled by~\cref{lem:R-nodes}, are $(1, 1, 1, 1, 1, 1)$, $(1, 1, 1, 1, 1)$, $(2, 1, 1, 1)$, $(3, 1, 1)$ and~$(3, 1, 1, 1)$. By~\cref{lem:311} the case of a $(3,1 ,1)$ distribution vector does not have to be considered and thus by applying~\cref{lem:(11111),lem:(3111)unconstrainedPQ,lem:(2111)} exhaustively, we either reject or only the last case remains; see~\cref{fig:table}.
Finally, by~\cref{lem:(3111)P3}, a vertex~$v$ with alternation constraint that appears in a $P$-node with distribution vector $(3, 1, 1, 1)$ has a $K_3$-constraint.
\end{proof}

Hence, we may assume that only $K_3$-constraints occur and that every vertex~$v$ with $K_3$-constraint appears in the skeleton of a $P$-node~$\nu$ in~$T$ with distribution vector~$(3, 1, 1, 1)$, whose pole distinct from~$v$ also has a~$K_3$-constraint. This concludes Step~1.

Now move to Step~2.
Let~$v$ be a vertex with~$K_3$-constraint.
The three edges from~$E(v)$ contained in the same virtual edge in the skeleton of~$\nu$ must have pairwise distinct colors; otherwise,~$H$ is a no-instance.
Since there are no two adjacent $P$-nodes in an SPQR-tree and by assumption no vertex with alternation constraint appears in an $R$-node,~$v$ also appears in an $S$-node with distribution vector~$(3, 3)$. 
Let~$\mu$ be an $S$-node in~$T$ that contains~$v$.
Since there are no two adjacent~$S$-nodes in an SPQR-tree, for each neighbor~$u$ of~$v$ in the skeleton of~$\mu$, there is a $P$-node adjacent to~$\mu$ in~$T$ with poles~$v$,~$u$. 
Thus, by assumption, the neighbors of~$v$ in the skeleton of~$\mu$ also have a $K_3$-constraint.
Observe that the distribution vector of the neighbors of~$v$ in~$\mu$ 
cannot be $(5, 1)$ as all neighbors of~$\mu$ are $R$- or $P$-nodes; i.e., the vertices in the corresponding skeletons have degree at least~3. 
Hence the neighbors of~$v$ in~$\mu$ also have distribution vectors $(3, 3)$.
Iteratively, it follows that every vertex in the skeleton of~$\mu$ has a $K_3$-constraint and shares a $P$-node with each of its two neighbors.

Consider an $S$-node $\mu$ in $T$ that contains alternation-constrained vertices $v_0, \dots, v_{k-1}$ in this order; see \cref{fig:S-node}. 
In the following, we consider the indices of the vertices and edges in $\mu$ modulo~$k$.
For every~$0 \leq i < k$, we denote the virtual edge between $v_i$ and~$v_{i+1}$ by~$e_i$ and
let~$\nu_i$ be the $P$-node adjacent to~$\mu$ in $T$ with poles~$v_i, v_{i+1}$.
Note that for every~$i$, the virtual edge~$e$ in the skeleton of $\nu_i$ that contains three edges from $E(v_i)$ also contains three edges from $E(v_{i+1})$, 
since~$e$ is the virtual edge representing~$\mu$.
Thus, if we fix the order of the three edges from $E(v_i)$ in~$e_i$, this fixes
the order of the three edges from $E(v_{i+1})$ in $e_i$.
Since~$v_{i+1}$ has an alternation constraint, this also fixes the order of the edges from $E(v_{i+1})$ in $e_{i+1}$. 
In this way, a fixed order of the three edges from $E(v_1)$ in $e_1$ implies an order of the edges from $E(v_{k-1})$ in~$e_{k-1}$, which in turn implies
an order on the three edges from $E(v_1)$ in~$e_{k-1}$.
If there exists an order of the three edges from $E(v_1)$ in $e_1$ that implies an order of the remaining edges from~$E(v_1)$ in~$e_{k-1}$ such that the $K_3$-constraint is satisfied, we obtain an equivalent instance by removing all $K_3$-constraints of vertices in the skeleton of~$\mu$, since we can reorder the parallels adjacent to~$\nu$ independently of the remaining graph.
Otherwise,~$H$ is a no-instance.
By the discussion above, we have the following.

\begin{figure}
\centering
    \includegraphics{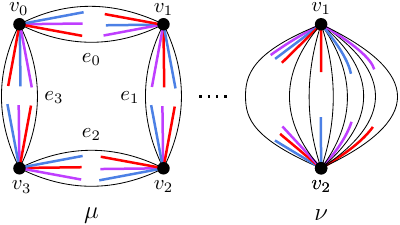}
  \caption{An $S$-node $\mu$ and an adjacent $P$-node~$\nu$.}
  \label{fig:S-node}
\end{figure}

\begin{lemma}
\label{lem:S-nodes}
    Let $\mu$ be an $S$-node in $T$ that contains vertices with  $K_3$-constraint. There is an $O(\deg(\mu))$-algorithm that either 
    recognizes that~$H$ is a no-instance, or computes an equivalent instance by removing all $K_3$-constraints of vertices in the skeleton of~$\mu$.
\end{lemma}

    Now we may assume that our graph~$H$ does not contain alternation constraints and start with Step~3.  To solve such an instance, we expand each vertex with its associated PQ-tree, if any, into a gadget that allows the same circular orders of its incident edges as the PQ-tree (essentially a P-node becomes a normal vertex, whereas a Q-node expands into a wheel as described in~\cite{DBLP:journals/jgaa/GutwengerKM08}).  Embedding the resulting graph~$H^\star$ and then contracting the gadgets back into single vertices already ensures that for each vertex of $H$ the order of its incident edges is represented by its PQ-tree.  Since our synchronized PQ-trees only involve Q-nodes, it then suffices to ensure that synchronized Q-nodes are flipped consistently.  To this end, observe that each wheel is 3-connected and hence its embedding is determined by a single R-node in the SPQR-tree of~$H^\star$.  This allows us to express such constraints in a simple 2-SAT formula of linear size, similar to, e.g.,~\cite{DBLP:journals/talg/BlasiusFR23,DBLP:journals/jgaa/GutwengerKM08}. Therefore, we obtain the following.

\begin{theorem}
\label{th:main-theorem}
Let~$A=(G,\Chi)$ be an~$n$-vertex AT-graph such that~$\lcc(A) \leq 3$. There exists an~$O(n)$-time algorithm that decides whether~$A$ is a positive instance of \SATor and that, in the positive case, computes a simple realization of $A$.
\end{theorem}

\begin{proof}
By \cref{lem:reduction-realizability-acp,lem:reduction-to-2-connected}, AT-graph realizability with~$\lcc(A) \leq 3$ can be reduced to \textsc{2-connected}~\acpG in linear time.
Each of the resulting instances~$H$ of~\textsc{2-connected}~\acpG we solve independently as follows.
First we compute an equivalent instance $H'$ without alternation constraints by applying \cref{lem:consecConstraints,lem:(11111),lem:(3111)P3,lem:(3111)unconstrainedPQ,lem:(2111),lem:S-nodes,lem:only3111left} exhaustively (and reject if necessary).
Since in~$H'$, no vertex has an alternation constraint, a planar embedding of~$H'$ is feasible if and only if it satisfies the PQ-constraints.
To solve such an instance, we expand each vertex with its associated PQ-tree, if any, into a gadget that allows the same circular orders of its incident edges as the PQ-tree (essentially a P-node becomes a normal vertex, whereas a Q-node expands into a wheel as described in~\cite{DBLP:journals/jgaa/GutwengerKM08}).  Embedding the resulting graph~$H^\star$ and then contracting the gadgets back into single vertices already ensures that for each vertex of $H$ the order of its incident edges is represented by its PQ-tree.  Since our synchronized PQ-trees only involve Q-nodes, it then suffices to ensure that synchronized Q-nodes are flipped consistently.  To this end, observe that each wheel is 3-connected and hence its embedding is determined by a single R-node in the SPQR-tree of~$H^\star$.  This allows us to express such constraints in a simple 2-SAT formula of linear size, similar to, e.g.,~\cite{DBLP:journals/talg/BlasiusFR23,DBLP:journals/jgaa/GutwengerKM08}. 
If~$H'$ is a yes-instance, such a procedure also provides a feasible embedding of $H'$, which can be obtained from an embedding of $H^\star$ that satisfies all the synchronization constraints by contracting each gadget to a single vertex.
Starting from the embedding of $H'$, by selecting suitable permutations for the children of the P-nodes (according to \cref{lem:consecConstraints,lem:(11111),lem:(3111)P3,lem:(3111)unconstrainedPQ,lem:S-nodes,lem:(2111)}), we can obtain a feasible embedding of $H$. Similarly, feasible embeddings of each block of the input auxiliary graph can be combined, by reordering the blocks around common cut vertices (according to \cref{lem:reduction-to-2-connected}), to obtain a feasible embedding $\Gamma$ of it.
To get a realization of~$A$ from~$\Gamma$,
we delete all crossing vertices and reinsert for every connected component~$X$ of the crossing graph~$\mathcal{C}(A)$, the edges in $E(X)$.

Clearly, applying \cref{lem:consecConstraints,lem:(11111),lem:(3111)P3,lem:(2111),lem:only3111left} exhaustively takes linear time in total.
Also applying \cref{lem:(3111)unconstrainedPQ} exhaustively can be done efficiently, if
we store all vertices that initially satisfy the conditions of the lemma in a list~$L$. When the alternation constraint of a vertex~$v$ in $L$ is removed, we remove $v$ from $L$ and add all vertices to~$L$ that have an alternation constraint and appear in a $P$-node together with~$v$.
This way every vertex is processed only once.
For applying \cref{lem:S-nodes} exhaustively, we iterate over all $S$-nodes in~$T$. For each $S$-node we need time linear in its degree.
This yields a linear running time in total.
\end{proof}

\section{Conclusions and Open Problems}\label{se:conclusions}

We proved that deciding whether an AT-graph $A$ is simply realizable is \NP-complete, already when the largest connected components of the crossing graph ${\cal C}(A)$ have 6 vertices only. On the other hand, we described an optimal linear-time algorithm that solves the problem when the largest components of ${\cal C}(A)$ have size at most $3$. This is the first efficient algorithm for the \SATRfull problem that works on general graphs.

An open problem that naturally arises from our findings is filling the gap between tractability and intractability: What is the complexity of \SATRfull if~$\lambda(A)$ is 4 or 5?
A first issue is that \Cref{lem:untangling} only allows to untangle cliques of size~3 and it is not clear whether a similar result can be proved for components of size~4.  Furthermore, contracting larger crossing structures yields more complicated alternation constraints and it is not clear whether they can be turned into PQ-trees, similar to the case of components of size~3. We therefore feel that different techniques may be necessary to tackle the cases where~$4 \le \lambda(A) \le 5$. 

Another interesting direction is to study alternative structural parameters under which the problem can be tackled, and which are not ruled out by our hardness result, as discussed in the introduction; for example the vertex cover number of ${\cal C}(A)$. 
Finally, one can try to extend our approach to the ``weak'' setting (i.e., the \WATRfull problem), still requiring a simple realization.

\bibliographystyle{plainurl}
\bibliography{ref}

\begin{thebibliography}{10}

\bibitem{DBLP:journals/talg/BlasiusFR23}
Thomas Bl{\"{a}}sius, Simon~D. Fink, and Ignaz Rutter.
\newblock Synchronized planarity with applications to constrained planarity
  problems.
\newblock {\em {ACM} Trans. Algorithms}, 19(4):34:1--34:23, 2023.
\newblock \href {https://doi.org/10.1145/3607474} {\path{doi:10.1145/3607474}}.

\bibitem{BlasiusR16}
Thomas Bl{\"{a}}sius and Ignaz Rutter.
\newblock Simultaneous {PQ}-ordering with applications to constrained embedding
  problems.
\newblock {\em {ACM} Trans. Algorithms}, 12(2):16:1--16:46, 2016.
\newblock \href {https://doi.org/10.1145/2738054} {\path{doi:10.1145/2738054}}.

\bibitem{DBLP:journals/jcss/BoothL76}
Kellogg~S. Booth and George~S. Lueker.
\newblock Testing for the consecutive ones property, interval graphs, and graph
  planarity using {PQ}-tree algorithms.
\newblock {\em J. Comput. Syst. Sci.}, 13(3):335--379, 1976.
\newblock \href {https://doi.org/10.1016/S0022-0000(76)80045-1}
  {\path{doi:10.1016/S0022-0000(76)80045-1}}.

\bibitem{booth1975pq}
Kellogg~Speed Booth.
\newblock {\em PQ-tree algorithms.}
\newblock University of California, Berkeley, 1975.

\bibitem{DBLP:books/ph/BattistaETT99}
Giuseppe {Di Battista}, Peter Eades, Roberto Tamassia, and Ioannis~G. Tollis.
\newblock {\em Graph Drawing: Algorithms for the Visualization of Graphs}.
\newblock Prentice-Hall, 1999.

\bibitem{DBLP:journals/algorithmica/BattistaT96}
Giuseppe {Di Battista} and Roberto Tamassia.
\newblock On-line maintenance of triconnected components with {SPQR}-trees.
\newblock {\em Algorithmica}, 15(4):302--318, 1996.
\newblock \href {https://doi.org/10.1007/BF01961541}
  {\path{doi:10.1007/BF01961541}}.

\bibitem{dt-opl-96}
Giuseppe {Di Battista} and Roberto Tamassia.
\newblock On-line planarity testing.
\newblock {\em {SIAM} J. Comput.}, 25(5):956--997, 1996.

\bibitem{DBLP:journals/csur/DidimoLM19}
Walter Didimo, Giuseppe Liotta, and Fabrizio Montecchiani.
\newblock A survey on graph drawing beyond planarity.
\newblock {\em {ACM} Comput. Surv.}, 52(1):4:1--4:37, 2019.
\newblock \href {https://doi.org/10.1145/3301281} {\path{doi:10.1145/3301281}}.

\bibitem{DBLP:books/daglib/0030488}
Reinhard Diestel.
\newblock {\em Graph Theory, 4th Edition}, volume 173 of {\em Graduate texts in
  mathematics}.
\newblock Springer, 2012.

\bibitem{DBLP:conf/wg/GassnerJPSS06}
Elisabeth Gassner, Michael J{\"{u}}nger, Merijam Percan, Marcus Schaefer, and
  Michael Schulz.
\newblock Simultaneous graph embeddings with fixed edges.
\newblock In Fedor~V. Fomin, editor, {\em 32nd International Workshop on
  Graph-Theoretic Concepts in Computer Science, {WG} 2006}, volume 4271 of {\em
  Lecture Notes in Computer Science}, pages 325--335. Springer, 2006.
\newblock \href {https://doi.org/10.1007/11917496\_29}
  {\path{doi:10.1007/11917496\_29}}.

\bibitem{DBLP:reference/cg/2004}
Jacob~E. Goodman and Joseph O'Rourke, editors.
\newblock {\em Handbook of Discrete and Computational Geometry, Second
  Edition}.
\newblock Chapman and Hall/CRC, 2004.
\newblock \href {https://doi.org/10.1201/9781420035315}
  {\path{doi:10.1201/9781420035315}}.

\bibitem{DBLP:journals/jgaa/GutwengerKM08}
Carsten Gutwenger, Karsten Klein, and Petra Mutzel.
\newblock Planarity testing and optimal edge insertion with embedding
  constraints.
\newblock {\em J. Graph Algorithms Appl.}, 12(1):73--95, 2008.
\newblock \href {https://doi.org/10.7155/JGAA.00160}
  {\path{doi:10.7155/JGAA.00160}}.

\bibitem{DBLP:conf/gd/GutwengerM00}
Carsten Gutwenger and Petra Mutzel.
\newblock A linear time implementation of {SPQR}-trees.
\newblock In Joe Marks, editor, {\em 8th International Symposium on Graph
  Drawing, {GD} 2000, Proceedings}, volume 1984 of {\em LNCS}, pages 77--90.
  Springer, 2000.
\newblock \href {https://doi.org/10.1007/3-540-44541-2\_8}
  {\path{doi:10.1007/3-540-44541-2\_8}}.

\bibitem{DBLP:books/sp/20/Hong20}
Seok{-}Hee Hong.
\newblock Beyond planar graphs: Introduction.
\newblock In Seok{-}Hee Hong and Takeshi Tokuyama, editors, {\em Beyond Planar
  Graphs, Communications of {NII} Shonan Meetings}, pages 1--9. Springer, 2020.
\newblock \href {https://doi.org/10.1007/978-981-15-6533-5\_1}
  {\path{doi:10.1007/978-981-15-6533-5\_1}}.

\bibitem{DBLP:journals/jacm/HopcroftT74}
John~E. Hopcroft and Robert~Endre Tarjan.
\newblock Efficient planarity testing.
\newblock {\em J. {ACM}}, 21(4):549--568, 1974.
\newblock \href {https://doi.org/10.1145/321850.321852}
  {\path{doi:10.1145/321850.321852}}.

\bibitem{DBLP:journals/jcss/ImpagliazzoPZ01}
Russell Impagliazzo, Ramamohan Paturi, and Francis Zane.
\newblock Which problems have strongly exponential complexity?
\newblock {\em J. Comput. Syst. Sci.}, 63(4):512--530, 2001.
\newblock \href {https://doi.org/10.1006/jcss.2001.1774}
  {\path{doi:10.1006/jcss.2001.1774}}.

\bibitem{DBLP:journals/jct/Kratochvil91a}
Jan Kratochv{\'{\i}}l.
\newblock String graphs. {II.} {R}ecognizing string graphs is {NP}-hard.
\newblock {\em J. Comb. Theory, Ser. {B}}, 52(1):67--78, 1991.

\bibitem{k-spspc-94}
Jan Kratochv{\'{\i}}l.
\newblock A special planar satisfiability problem and a consequence of its
  {NP}-completeness.
\newblock {\em Discrete Applied Mathematics}, 52(3):233 -- 252, 1994.
\newblock \href {https://doi.org/10.1016/0166-218X(94)90143-0}
  {\path{doi:10.1016/0166-218X(94)90143-0}}.

\bibitem{DBLP:journals/siamdm/KratochvilLN91}
Jan Kratochv{\'{\i}}l, Anna Lubiw, and Jaroslav Nesetril.
\newblock Noncrossing subgraphs in topological layouts.
\newblock {\em {SIAM} J. Discret. Math.}, 4(2):223--244, 1991.

\bibitem{Kratochvil1989}
Jan Kratochv{\' {\i}}l and Ji{\v r}{\'{\i}} Matou{\v s}ek.
\newblock {NP}-hardness results for intersection graphs.
\newblock {\em Commentationes Mathematicae Universitatis Carolinae},
  30(4):761--773, 1989.
\newblock URL: \url{http://eudml.org/doc/17790}.

\bibitem{DBLP:journals/jct/KratochvilM94}
Jan Kratochv{\'{\i}}l and Jir{\'{\i}} Matousek.
\newblock Intersection graphs of segments.
\newblock {\em J. Comb. Theory, Ser. {B}}, 62(2):289--315, 1994.

\bibitem{DBLP:journals/dcg/Kyncl11}
Jan Kyncl.
\newblock Simple realizability of complete abstract topological graphs in {P}.
\newblock {\em Discret. Comput. Geom.}, 45(3):383--399, 2011.

\bibitem{DBLP:journals/dcg/Kyncl20}
Jan Kyncl.
\newblock Simple realizability of complete abstract topological graphs
  simplified.
\newblock {\em Discret. Comput. Geom.}, 64(1):1--27, 2020.

\bibitem{DBLP:journals/siamcomp/Lichtenstein82}
David Lichtenstein.
\newblock Planar formulae and their uses.
\newblock {\em {SIAM} J. Comput.}, 11(2):329--343, 1982.
\newblock \href {https://doi.org/10.1137/0211025} {\path{doi:10.1137/0211025}}.

\bibitem{Menger1927}
Karl Menger.
\newblock Zur allgemeinen {K}urventheorie.
\newblock {\em Fundamenta Mathematicae}, 10(1):96--115, 1927.
\newblock URL: \url{http://eudml.org/doc/211191}.

\bibitem{Schaefer13}
Marcus Schaefer.
\newblock Toward a theory of planarity: Hanani-tutte and planarity variants.
\newblock {\em J. Graph Algorithms Appl.}, 17(4):367--440, 2013.
\newblock \href {https://doi.org/10.7155/JGAA.00298}
  {\path{doi:10.7155/JGAA.00298}}.

\bibitem{sss-rsgnp-03}
Marcus Schaefer, Eric Sedgwick, and Daniel Stefankovic.
\newblock Recognizing string graphs in {NP}.
\newblock {\em J. Comput. Syst. Sci.}, 67(2):365--380, 2003.

\bibitem{DBLP:journals/jcss/SchaeferS04}
Marcus Schaefer and Daniel Stefankovic.
\newblock Decidability of string graphs.
\newblock {\em J. Comput. Syst. Sci.}, 68(2):319--334, 2004.
\newblock \href {https://doi.org/10.1016/J.JCSS.2003.07.002}
  {\path{doi:10.1016/J.JCSS.2003.07.002}}.

\bibitem{DBLP:reference/crc/StreinuBPDCKAF99}
Ileana Streinu, K{\'{a}}roly Bezdek, J{\'{a}}nos Pach, Tamal~K. Dey, Jianer
  Chen, Dina Kravets, Nancy~M. Amato, and W.~Randolph Franklin.
\newblock Discrete and computational geometry.
\newblock In Kenneth~H. Rosen, John~G. Michaels, Jonathan~L. Gross, Jerrold~W.
  Grossman, and Douglas~R. Shier, editors, {\em Handbook of Discrete and
  Combinatorial Mathematics}. {CRC} Press, 1999.
\newblock \href {https://doi.org/10.1201/9781439832905.CH13}
  {\path{doi:10.1201/9781439832905.CH13}}.

\end{thebibliography}
\end{document}